\documentclass[11pt]{llncs}
\usepackage[margin=1in]{geometry}

\usepackage[ruled]{algorithm2e} 

\SetAlFnt{\small}
\SetAlCapFnt{\small}
\SetAlCapNameFnt{\small}
\SetAlCapHSkip{0pt}
\IncMargin{-\parindent}

\usepackage{amsmath, amssymb, amsfonts}
\usepackage{amstext}
\usepackage{graphicx}
\usepackage{color}
\usepackage{natbib}
\usepackage{bm, bbm}
\usepackage{gensymb}

\newcommand{\E}{\mathbb{E}}
\newcommand{\R}{\mathbb{R}}
\newcommand{\vv}[1]{\bm{#1}}

\newcommand{\setbvals}{V}
\newcommand{\setprices}{\mathcal{P}}
\newcommand{\Dmech}{\mathcal{M}_d}

\newcommand{\Xv}{\bm{X}}
\newcommand{\Pv}{\bm{P}}

\newcommand{\tv}{\bm{t}}

\newcommand{\Ger}[1]{{\color{red}#1}}
\newcommand{\evn}[1]{#1}

\begin{document}
\title{Optimal Mechanism Design with Risk-loving Agents}
\author{Evdokia Nikolova \and Emmanouil Pountourakis\and and Ger Yang}
\institute{The University of Texas at Austin}

\maketitle
	
\begin{abstract}
One of the most celebrated results in mechanism design is Myerson's characterization of the revenue optimal auction for selling a single item. However, this result relies heavily on the assumption that buyers are indifferent to risk. In this paper we investigate the case where the buyers are risk-loving, i.e. they prefer gambling to being rewarded deterministically. We use the standard model for risk from expected utility theory, where risk-loving behavior is represented by a convex utility function. 

We focus our attention on the special case of exponential utility functions. 
We characterize the optimal auction and show that randomization can be used to extract more revenue than when buyers are risk-neutral. Most importantly, we show that the optimal auction is simple: the optimal revenue can be extracted using a randomized take-it-or-leave-it price for a single buyer and using a loser-pay auction, a variant of the \evn{all-pay} auction, for multiple buyers. Finally, we show that these results no longer hold for convex utility functions beyond exponential. 
\end{abstract}

\section{Introduction}

The classic mechanism design problem, pioneered by Myerson's seminal work \citep{myerson1981optimal}, considers designing an auction that maximizes the auctioneer's revenue.  There is rich literature on this mechanism design problem under different settings.  However, most prior work assumes the buyers are utility maximizers with quasilinear utility functions, where the utility function is linear in either the payment or the buyer's value.  These assumptions often make the problem simple and easy to analyze.  In the real world, agents need not follow such assumptions.  In fact, under different behavioral models for the buyers, the auctioneer is able to draw more revenue than under the standard setting, where the buyers are maximizing their linear utility functions.  One particular example is when the buyers are risk-averse \citep{maskin1984optimal,chawla2017aversion}.  
In this case, the seller/auctioneer can design an ``insurance''-based auction to extract more revenue from risk-averse buyers.  
In this paper, we study the setting where the buyers are \emph{risk-loving}.  We ask whether the auctioneer can take advantage of such risk-loving behavior, and if so, what can be achieved?

Recently, experiments in electricity markets and transportation networks have demonstrated the importance of designing a mechanism for risk-loving agents.  
Electric utility companies are considering how to incentivize 
customers to reduce their electricity consumption in peak load times so as to alleviate the strain on the grid and to prevent expensive line capacity and transformer upgrades.  Some of the more successful attempts to achieve a desired ``demand response" have included offering lottery coupons to consumers for scaling back demand~\citep{li2015energy}.  
In transportation networks, similar lottery schemes have been applied to reduce congestion in the rush hour \citep{lu2015framework,balaji1,balaji2}.  
In both cases, more consumer response was elicited from lotteries, where a consumer was offered a small chance to win a big reward, than from small fixed payments.  
Hence, there is a need for a theoretical foundation and analysis of the optimal lottery schemes to improve the consumer response and experience in these nation critical infrastructure applications.

In economics, von Neumann-Morgenstern's expected utility theory \citep{von1945theory} has been a standard model to describe people's preferences. According to this theory, an agent evaluates the payoff of an event by applying a utility function on the wealth it generates, and takes the expectation over all possible events to evaluate the payoff of a given action.
As such, expected utility theory provides a simple way to describe how people behave when facing \emph{risk} --- a risk-averse player has a concave utility function, whereas a risk-loving player has a convex utility function.  
Consider a payment scheme where a buyer can choose one of two payment options.  In the first option, the buyer either pays $\$100$ or $\$0$, each with probability $50\%$.  In the second option, she has to pay $\$50$ with certainty. These two options have the same expected gains.  A risk-neutral buyer, who has a linear utility function, is indifferent between these two options.  A risk-averse buyer is going to choose the second option because she prefers the less risky payment scheme.  A risk-loving buyer will choose the first option because she is more willing to take risks.

In the above example, the expected payment the seller receives is $\$50$.  If the buyer is risk-loving, we can extract more revenue by replacing the first option with an even more risky payment option.  For example, we can offer another payment option in which the buyer pays $\$110$ or $0$, each with probability $50\%$.  From the risk-loving buyer's perspective, this new payment option is still preferable to the second option.  Therefore, the expected payment the seller receives will increase to $\$55$.  In fact, it has been shown by \citet{hinnosaar2016impossibility} that in the absence of any regulation, the seller is able to extract infinite expected revenue from a risk-loving buyer by simply taking advantage of this trick --- offering a menu option that asks the buyer to pay a very high amount with a very small probability.  Therefore, in this paper, we will mainly focus on the \emph{bounded transfer} setting, i.e. where we upper bound the \emph{ex-post} payment that the seller may ask the buyer to pay.  In other words, the amount of payment by the buyer is upper-bounded by some specific value under all circumstances. Particlarly, the \emph{bounded transfer} requirement can be shown to be equivalent to the buyers having a publicly known really high yet still bounded budget.

\subsection{Our results and Techniques}

In this paper, we focus on a special case of risk-loving agents, that use an exponential utility function of the form $u(x)=\beta(e^{\alpha x}-1)$. We seek to design individually rational and incentive compatible mechanisms that maximize the revenue.  We assume bounded transfers, that is the maximum payment of the mechanism is bounded, and characterize the optimal mechanism. Surprisingly, we show that if the value distribution of the agents is well behaved, then the optimal revenue can be extracted using a randomized take-it-or-leave-it price for a single buyer and a loser-pay auction, a variant of the all-pay auction where the winner gets a refund, for multiple buyers. 

Our analysis combines a generalized virtual value function similarly to Myerson's analysis~\citep{myerson1981optimal} and the duality framework developed by~\citet{cai2016duality}. \evn{In particular}, we upper bound the revenue of the optimal mechanism by defining a dual solution that can be interpreted as a generalization of the virtual value function. Then we show that this solution matches the revenue obtained by a randomized-take-it-or-leave-it price and the loser-pay auction, for a single buyer and for multiple buyers, respectively. To our surprise, the virtual value function that captures the marginal revenue is different in the single buyer and multiple buyer settings, which may be explained due to the additional uncertainty introduced by the extra buyers. 

These results are in stark contrast with the risk-averse setting where the seller can improve the revenue by offering a plethora of lotteries each with a deterministic price but different allocation probabilities \citep{maskin1984optimal}. The risk-averse buyer opts to pay for lotteries that are priced close to her value and the risk is used as a deterrent for under-bidding. On the other hand, we can extract more revenue from a risk-loving buyer by randomizing the payment. This is because the buyer gains more utility from gambles so that we can increase the probability that the price is accepted. This difference in how risk behavior is exploited explains the conceptually different nature of revenue maximizing mechanisms in the two settings.

\section{Related Work}
Most work \evn{on optimal mechanism design} beyond the risk-neutral setting has focused on risk-averse preferences. 
The classic results of \citet{maskin1984optimal} and \citet{matthews1983selling} provide a characterization of the optimal mechanism with concave utility functions. A recent result in this area by \citet{dughmi2012mechanisms} is that any mechanism designed for risk-neutral buyers can be adjusted to also align the incentives of risk-averse buyers and obtain similar guarantees.  \citet{fu2013prior} consider the design of prior-independent mechanisms (that have no access to the buyers' private value distributions) for risk-averse buyers. Finally,  \citet{chawla2017aversion} study the design of robust mechanisms under the cumulative prospect theory model.

To the best our knowledge, the only work on mechanism design under risk-loving behavior is by \citet{hinnosaar2016impossibility}, who shows that in the absence of regulations, the seller can extract infinite revenue from the buyer with asymptotically risk-loving behavior under both \evn{the} expected utility theory and prospect theory models.

Recently, the duality theory framework has drawn attention in the mechanism design community for understanding optimal mechanisms for selling multiple items.  For example, \citet{daskalakis2017strong,daskalakis2013mechanism} and \citet{giannakopoulos2014duality,giannakopoulos2015selling} discovered the connection between the dual problem and the optimal transport (bipartite matching) problem.    \citet{cai2016duality} consider a duality framework via linear programming, and identify a connection between the virtual valuations and the dual variables.  In our setting, the problem results in a different form of dual problem than in the multi-item setting, hence we seek to establish a new duality framework that diverges from the multi-item setting to different behavior models.

\section{Problem Statement} \label{sec:mech:ps}
We study revenue maximization for a single seller and $n$ symmetric buyers. The seller has a single item to sell and each buyer $i$ has a private value $t_i$ for the item. We use $\tv=(t_1,\dots,t_n)$ to denote the values of all buyers.  We let $\setbvals=\{v_1,v_2,\dots,v_K\}$ denote the set of all possible values, which is shared by all buyers. For simplicity, we assume $v_1=0$ and $v_1<v_2<\dots<v_K$.  Additionally, we assume each buyer's private value is drawn independently from a known identical distribution with probability mass function $f$.  Without loss of generality, we assume $f(v)>0$ for all $v \in \setbvals$.  Further, we let $\setprices = \{z_1,z_2,\dots,z_M\}$ denote the set of allowed payments, where $z_1=0$ (no positive transfers) and $z_1 < z_2 < \dots <z_M$. Here we implicitly assume that $\setprices$ is upper-bounded by $z_M$ \footnote{Without this assumption, it can be shown that there exists a mechanism that attains infinite revenue from risk-loving buyers \cite{hinnosaar2016impossibility}.}. This implies that our setting becomes equivalent with the case where the payments are unconstrainted but the buyer has a publicly known budget of $z_M$ as we show in subsection~\ref{subsec:budget}. We additionally require that $z_M > v_K$, that is, the upper bound of the payment is larger than the largest possible buyer's value.


Each buyer seeks to maximize her utility given by a function $u: \R \rightarrow \R$, which we assume is strictly increasing and $u(0)=0$. If $u$ is linear, then we say the buyers are \emph{risk-neutral} and if $u$ is convex, then we say that the buyers are \emph{risk-loving}.  For the rest of the paper, we focus on a special case of convex utility, specifically the exponential utility function  given by $u(x)=\beta(e^{\alpha x}-1)$ for some $\alpha>0$ and $\beta>0$. Unless otherwise noted, we will assume such \evn{an} exponential utility function for the buyers.

\textbf{Notation.} Let $[R]$ denote the set $\{1,2,\dots,R\}$, for any positive integer $R$. For any vector $\vv{v}$, we use $\vv{v}_{-i}$ to denote the vector generated by removing the $i$-th coordinate from $\vv{v}$.  Also, we use $(v,\vv{v}_{-i})$ to denote the vector generated by replacing the $i$-th coordinate of $\vv{v}$ with $v$.

\subsection{Direct Mechanisms and Bayesian Incentive Compatibility}
In a direct mechanism the auctioneer elicits bids from each buyer and then decides on their allocation probabilities and payments. 
We represent such a mechanism by $\Dmech=(\Xv,\Pv)$, where $\Xv: \setbvals^n \rightarrow \{0,1\}^n$ is a random allocation function and   $\Pv: \setbvals^n \rightarrow \setprices^n$ is a payment function which can also be randomized. 
Given all buyers' values \evn{$\vv{t}=(t_1,\dots,t_n)$},  we refer to  the random variable \evn{$(\Xv(\vv{t}),\Pv(\vv{t}))$} as the \emph{outcome} of the mechanism at \evn{$\vv{t}$}.

We require that our mechanism $\Dmech$ is \emph{Bayesian incentive compatible} (BIC), that is, for each buyer, it is in her best interest to truthfully report her value in expectation. Note that this expectation takes into account the randomness of the mechanism as well as the uncertainty about the other buyers' values.  Formally, $\Dmech = (\Xv,\Pv)$ is BIC if for any $i \in [n]$ and for any $v \in \setbvals$ and $v' \in \setbvals$, it holds that

\begin{equation} \label{eq:ic}
\E[u(v \Xv(v,\vv{t}_{-i}) - \Pv(v,\vv{t}_{-i}))] \ge \E[u(v \Xv(v',\vv{t}_{-i}) - \Pv(v',\vv{t}_{-i}))],
\end{equation} 
where the expectation is taken over $\Xv$, $\Pv$, and \evn{$\vv{t}_{-i}$}.  A mechanism is \emph{individually rational} (IR) if it guarantees a non-negative expected utility for every buyer that truthfully reveals her value, i.e, for any $i \in [n]$ and for any $v \in \setbvals$, it holds that

\begin{equation} \label{eq:ir}
\E[u(v \Xv(v,\vv{t}_{-i}) - \Pv(v,\vv{t}_{-i}))] \ge 0.
\end{equation}
Note that if we only allow non-negative payments, then we must have \evn{$\Pv(0,\vv{t}_{-i}) \le 0$} almost surely.

\subsection{Bounded transfers and budgeted buyers}\label{subsec:budget}

As we stated in the beginning of Seciton~\ref{sec:mech:ps} we require the mechanism to charge ex-post payments from a finite pool of $\setprices$, where $z_M$ is the largest ex-post price that also satisfies that  $z_M>v_K$, i.e., the upper-bound on the payment is larger than the highest value of the buyer. The finiteness of $z_M$ can be thought of as buyers having a finite budget equal to $z_M$. Particulalry, for the case where $\setprices$ was unbounded but the buyers had  a budget of $z_M$ then no reveue maximizing IR mechanism would ever charge an ex-post price larger than $z_M$. Similarly, in any feasible mechanism under the bounded-transfer setting, the buyers with budget greater than the upper-bound on the ex-post price behave as if they had no budget at all. As a result, in both of those cases the revenue-maximizing BIC and IR mechanism are the same.

\subsection{Myerson's Mechanism and Virtual Values}

One of the fundamental results of auction theory is Myerson's    characterization of revenue optimal mechanisms  for risk-neutral buyers \citep{myerson1981optimal}. This is achieved by an ammortized analysis that expresses the revenue of any mechanism via the \emph{virtual value function} $\phi(v)$, which captures the marginal revenue of allocating to a buyer with value $v$. The virtual value function is defined for a  continuous distribution of values (and can be similarly  defined for a discrete one), with cumulative distribution function $F$ and probability density function $f$, as \begin{equation} \label{eq:meyerson_virtual_value}  
	\phi(v) = v - \frac{1-F(v)}{f(v)}.
\end{equation}
The revenue of the mechanism equals the expected virtual surplus, i.e., the expected virtual value of the winner. As a result, if the value distributions satisfy certain properties, the optimal mechanism turns out to be quite simple: for a single buyer it is just a take-it-or-leave-it price and for multiple symmetric buyers it is the second price auction with a common reserve. However, this definition of virtual values heavily relies on the risk-neutrality assumption.  Our analysis generalizes this definition for risk-loving buyers in Definition~\ref{risklovingvvf} in order to derive our results. 

\subsection{Revenue Maximization as an Optimization Problem}

Our goal is to characterize the optimal mechanism for revenue maximization. To that end, we model the mechanism design question as an optimization problem. We define the decision variables 
$\{y_{i,j}^0, y_{i,j}^1\}_{i \in [n], j \in [M]}$,
where $y_{i,j}^0: \setbvals^n \rightarrow [0,1]$ and $y_{i,j}^1: \setbvals^n \rightarrow [0,1]$,  that encode the mechanism $\Dmech$ as follows:   $y_{i,j}^0(\tv)$ represents the probability that buyer $i$ does not get the item and pays $z_j$ when the buyers' values are $\tv$. Similarly, $y_{i,j}^1(\tv)$ represents the probability that buyer $i$ gets the item and pays $z_j$, given the buyers' values are $\tv$.

Those decision variables capture both the allocation and the payment of the mechanism given any reported values. To see this, the allocation probability that buyer $i$ gets the item given values $\tv$ is  
$\sum_{j} y_{i,j}^1(\tv)$ and the expectation of her randomized payment is  
$\sum_{j} z_j y_{i,j}^1(\tv)+\sum_j z_j y_{i,j}^0(\tv)$ where the first and second summand correspond to her expected payment if she wins or loses the item respectively. 

For the sake of succinctness of our optimization problem formulation, we further define the {\em interim} version of the decision variables $y_{i,j}^1(\tv),y_{i,j}^0(\tv)$, denoted by $y_{i,j}^1(v_k),y_{i,j}^0(v_k)$. Namely, given that the buyer has value $v_k$, what is the expected probability of winning/losing the item and paying value $z_j$ in expectation over the values of the other buyers \evn{$\vv{v}_{-i}$}? These interim variables are given by:
\begin{align} 
&y_{i,j}^1(v_k) = \sum_{\vv{v}_{-i} \in V^{n-1}} y_{i,j}^1(v_k, \vv{v}_{-i}) f(\vv{v}_{-i}),
&y_{i,j}^0(v_k) = \sum_{\vv{v}_{-i} \in V^{n-1}} y_{i,j}^0(v_k, \vv{v}_{-i}) f(\vv{v}_{-i}).
\label{eq:lp_cxij}
\end{align}

We can express the interim allocation $x_i(k)$ of buyer $i$ at value $v_k$ as $x_i(k) = \sum_j y_{i,j}^1(v_k)$ and her interim payments in case of win $p_i(k)$ and loss $q_i(k)$ as:
\begin{align*}
&p_i(k) = \frac{\sum_j z_j y_{i,j}^1(v_k)}{x_i(k)},
&q_i(k) = \frac{\sum_j z_j y_{i,j}^0 (v_k)}{1-x_i(k)}. 
\end{align*}
 The above follow from the definition of conditional probability. 
 With this notation, we can rewrite the BIC constraint as:
\begin{equation} \label{eq:lp_cbic}
\sum_j \left[ y_{i,j}^1(v_k) u(v_k - z_j) + y_{i,j}^0(v_{k}) u(-z_j) \right] \ge \sum_j \left[ y_{i,j}^1(v_{k'}) u(v_k - z_j) + y_{i,j}^0(v_{k'}) u(-z_j) \right], \forall k, k' \in [K],
\end{equation}
where the first and second summand on the left hand side correspond to the expected utility if buyer $i$ wins and loses\evn{, respectively,} after truthfully reporting $v_k$. Similarly, the first and second summand on the right hand side correspond to the expected utility if buyer $i$ wins and loses, respectively, after misreporting $v_{k'}$.

In addition, we can write the IR constraint as
\begin{equation} \label{eq:lp_cir}
\sum_j \left[ y_{i,j}^1(v_k) u(v_k - z_j) + y_{i,j}^0(v_k) u(-z_j) \right] \ge 0, \qquad \forall k \in [K].
\end{equation}

Finally, we need to satisfy the feasibility constraints

\begin{align} \label{eq:lp_cf}
&\sum_j \left( y_{i,j}^0(\vv{v}) + y_{i,j}^1(\vv{v}) \right) = 1, & \sum_i \sum_j y_{i,j}^1(\vv{v}) \le 1, \qquad \forall \vv{v} \in \setbvals^n.
\end{align}

Therefore, we can find the optimal mechanism by solving the following linear program:
\begin{align} 
\text{Maximize} \qquad &\sum_i \sum_{v \in \setbvals} f(v) \sum_j z_j \left[ y_{i,j}^0(v) +  y_{i,j}^1(v) \right] \label{eq:lp_obj} \\
\text{Subject to} \qquad & \text{Constraints \eqref{eq:lp_cxij}, \eqref{eq:lp_cbic}, \eqref{eq:lp_cir}, and \eqref{eq:lp_cf}.} \nonumber\\
& y_{i,j}^0(\vv{v}) \ge 0, \quad y_{i,j}^1(\vv{v}) \ge 0, & \forall \vv{v} \in \setbvals^n. \nonumber
\end{align}

\subsection{Overview of Main Theorems and Results}
The main result of this paper is that there is no need to actually solve the linear program \eqref{eq:lp_obj} in order to compute the optimal mechanism.  Instead, we take advantage of the linear program formulation \eqref{eq:lp_obj} of the problem to help us derive simple mechanisms that are optimal.  \evn{In particular}, when there is a single risk-loving buyer with \evn{an} exponential utility function, we show that the optimal mechanism is a \emph{randomized ``take-it-or-leave-it" price}, which offers the buyer a single randomized price irrespectively of her value.  We present this result in the following theorem:
\begin{theorem}[Restatement of Theorem~\ref{thm:mec_exputil}] \label{thm:main_single_buyer}
	Consider a single risk-loving buyer with exponential utility.  The optimal mechanism is the revenue maximizing randomized take-it-or-leave-it price.
\end{theorem}

When there are multiple symmetric risk-loving buyers, we show that 
the optimal mechanism is a \emph{loser-pay} auction with a reserve price.  In a loser-pay auction, the item is awarded to the buyer with the highest value but only the buyers who do not get the item are paying. Similarly to the single buyer case, all payments are randomized between the minimum and the maximum price.  We will show the following theorem:
\begin{theorem}[Restatement of Theorem~\ref{thm:two_buyer_regular}] \label{thm:main_multiple_buyers}
	Consider $n \ge 2$ risk-loving buyers with exponential utility $u(x)=\beta(e^{\alpha x}-1)$.  Assume  $z_M \gg \alpha$.  Then, the optimal mechanism is a loser-pay auction.
\end{theorem}

%

\section{Optimal Mechanism Design for a Single Buyer} \label{sec:single_buyer}

In this section, we characterize the revenue-maximizing mechanism for selling an item to a \evn{single} risk-loving agent. Specifically, we show that \evn{the optimal mechanism is} a \emph{randomized ``take-it-or-leave-it" price}, that offers the buyer a single randomized price irrespectively of her value. \evn{To prove that, we first} characterize the revenue generated by the optimal randomized take-it-or-leave-it price. Then, we prove that this mechanism \evn{remains} optimal if we allow an arbitrary \Ger{BIC} mechanism, by utilizing the optimization problem formulation and the duality framework to find a matching upper bound. Note that this result is quite similar to what the Myerson characterization implies for the single risk-neutral buyer setting with the exception that the optimal ``take-it-or-leave-it" price is always deterministic. However, as Example~\ref{ex:rl-incr-rev} demonstrates, this is no longer true for the risk-loving setting, and randomizing the price is required even in the case where the allocation probability is deterministic.  

For the rest of this section, we are going to use the following \emph{menu of options} interpretation that provides an equivalent description of BIC mechanisms in the single buyer case: A menu consists of a tuple $(\Xv_i,\Pv_i)$ where $\Xv_i$ and $\Pv_i$ represent the allocation and payment random variables. Instead of the buyer revealing her value, the seller offers the buyer a collection of menu options and the buyer chooses the menu option that maximizes her utility.  For example, a deterministic take-it-or-leave-it price can be described with the menu options $(0,0),(1,\Pv)$ where $\Pv$ is a point-mass at \evn{some} price $p$. 

\subsection{Sub-optimality of Deterministic ``Take-it-or-leave-it" Prices}

We illustrate that randomizing a take-it-or-leave-it price can result in increased revenue as the buyer shifts from a risk-neutral to a risk-loving utility function.
For the sake of the presentation we use a continuous distribution but similar results can be derived using a discrete value distribution.

\begin{example} \label{ex:rl-incr-rev}
Assume the buyer's value $t$ for the item is distributed according to the Uniform distribution $U(0,1)$.  Then, the optimal mechanism with a risk-neutral buyer is to offer the buyer a take-it-or-leave it price of $1/2$, producing a revenue of $1/4$. Now, consider the case of a risk-loving buyer with utility $u$. Her utility of  accepting this price is $u(t-1/2)$.  Now, consider a different scheme using a randomized take-it-or-leave-it price: with probability $1/2$ pay nothing and with probability $1/2$ pay $1$.  \Ger{Note that this scheme has the same expected payment as the first one.  However, the} expected utility of the buyer for this option is 
$ \frac{1}{2}u\left(t\right) + \frac{1}{2}u\left(t-1\right) $
and by Jensen's inequality, we get that
$ \frac{1}{2}u\left(t\right) + \frac{1}{2}u\left(t-1\right) > u\left(t - \frac{1}{2}\right), $
which indicates that the expected utility function for the randomized menu option is always above the utility function of the menu option $(1,1/2)$.  This means that the probability that a buyer accepts the randomized menu option is greater than $1/2$, i.e. $\Pr[\frac{1}{2}u(t)+\frac{1}{2}u(t-1) \ge 0] \ge \Pr[u(t-1/2)\ge 0]=1/2$. Therefore, offering the randomized menu option earns more revenue.
\end{example}

\subsection{Optimal Take-it-or-leave-it Randomized Price}

In Example~\ref{ex:rl-incr-rev}, we saw that randomizing the take-it-or-leave-it price increased the revenue extracted by a risk-loving buyer.  This gives rise to the question\evn{: What} is the  revenue maximizing randomized take-it-or-leave-it price? 

\begin{definition}
	A  randomized take-it-or-leave-it price with allocation and price $\Pv$ is a mechanism that contains only two menu options $(0,0)$ and $(1,\Pv)$.
\end{definition}

In a randomized take-it-or-leave-it price scheme, the seller posts a (possibly randomized) price $\Pv$ of the item, and asks the buyer to accept it or not.  If the buyer rejects the price, then her allocation is zero and she pays nothing.  According to \citet{myerson1981optimal}, we know that with a risk-neutral buyer, it is optimal to post $\Pv$ with the expectation of $\Pv$ equal to $\arg\max_v (v \Pr[t \ge v])$.  In contrast, for a risk-loving buyer we show that $\Pv$ must be randomized and specifically have positive support only for the maximum and minimum allowable price.  
Formally, we state this in the following theorem: 
\begin{theorem} \label{thm:rl-random}
	With a risk-loving buyer, the randomized take-it-or-leave-it price with following randomized payment rule
	$$
	\Pv = 
	\begin{cases}
	z_M, &\text{w.p. }  \frac{u(v^*)}{u(v^*)-u(v^*-z_M)}\\
	0, &\text{otherwise}
	\end{cases}
	$$ 
	where $v^* = \arg\max_{v \in \setbvals} \frac{z_M \Pr[t \ge v] u(v)}{u(v)-u(v-z_M)}$, is the optimal randomized take-it-or-leave-it price.  The optimal take-it-or-leave-it price has revenue $ \max_{v \in \setbvals} \frac{z_M \Pr[t \ge v] u(v)}{u(v)-u(v-z_M)}$. 
\end{theorem}

Note that for a linear utility function $u$ this implements the optimal deterministic take-it-or-leave-it price for a risk neutral agent and produces the exact same revenue.  We call the optimal take-it-or-leave-it mechanism the \emph{revenue maximizing randomized take-it-or-leave-it price}.

Theorem~\ref{thm:rl-random} can be proved in two steps.  In the first step, we apply Jensen's inequality to show that given any pricing rule $\Pv$, we can construct another pricing rule $\Pv'$ that randomizes between $0$ and $z_M$, and achieves a larger utility than $\Pv$ for any value $v \in \setbvals$.  In the second step, we use the individual rationality constraint to derive the optimal pricing.  Due to space constraints, we defer the full proof to Appendix~\ref{sec:app_single_buyer1}.

\subsection{Duality Theory for Optimal Mechanisms}
\label{sec:sb_exp1}

In this section, we prove our main theorem (Theorem~\ref{thm:single_buyer_regular}), namely that the revenue-maximizing take-it-or-leave-it randomized price is optimal among all possible mechanisms for revenue extraction. This is similar in nature to the risk-neutral setting where \citet{myerson1981optimal} showed that a take-it-or-leave-it payment is optimal.  It is important to note that the same is not true for the case of risk-averse agents where multiple menu options can be used to extract revenue that approaches the social welfare given sufficient aversion to risk. 

We prove this theorem by upper bounding the revenue of the optimal mechanism using the optimization problem formulation \eqref{eq:lp_obj} and employing the duality framework.  Specifically, we identify dual variables of the  Lagrangian dual program and show that  it matches the maximum revenue obtained by a take-it-or-leave-it randomized price.  The core idea is to define a \emph{virtual value function} that captures the marginal revenue and assume a regularity condition to close the duality gap. The virtual values used in interpreting Myerson's result heavily rely on the assumption of risk neutrality, therefore we need a new definition of virtual value.

\begin{definition}[Virtual value function for a single buyer]~\label{risklovingvvf}
	In the single buyer setting, the virtual value function $\phi_u: [K] \rightarrow \R$ with respect to utility function $u$ is defined as
	$$
	\phi_u(k) = \frac{1}{f(v_k)} \left( \sum_{k' \ge k}f(v_{k'}) \cdot \frac{u(v_k)}{u(v_k)-u(v_k-z_M)} - \sum_{k' \ge {k+1}}f(v_{k'}) \cdot \frac{u(v_{k+1})}{u(v_{k+1})-u(v_{k+1}-z_M)} \right).
	$$
\end{definition}

Note that this definition of the virtual value function reduces to Myerson's virtual value function for linear $u$. 
To see how this new definition of virtual value function is related to the marginal revenue, consider a take-it-or-leave-it mechanism with a pricing rule that asks the buyer to pay $z_M$ with probability $p$, and pay zero otherwise.  If we would like to guarantee that this pricing rule is accepted by a buyer with value greater than $v$, then by individual rationality, we can find that the largest $p$, the probability of paying $z_M$, we can set is $\frac{u(v)}{u(v)-u(v-z_M)}$.  Therefore, the expected revenue we have from this mechanism is $\sum_{v' \ge v} f(v') z_M \frac{u(v)}{u(v)-u(v-z_M)}$.
As a result, we can find that $f(v_k)\phi_u(k)z_M$ is indeed the marginal increase on the revenue by moving the threshold value from $v_{k+1}$ to $v_k$.
Given the definition of the virtual value function above, \evn{we need to define the following regularity condition}:
\begin{definition}[Regular distribution for a single buyer]
	In the single buyer setting, a distribution $f$ is regular if the corresponding virtual value function is monotone increasing, i.e. for any $k,k' \in [K]$ and $k > k'$, it holds that $\phi_u(k) > \phi_u(k')$.
\end{definition}

Note that with this regularity condition, in the optimal take-it-or-leave-it mechanism considered in Theorem~\ref{thm:rl-random}, the optimal quantile can also be found by looking at the smallest $k \in \{0,1,\dots,K\}$ such that $\phi_u(k)>0$. We are now ready to state our first main result in Theorem~\ref{thm:single_buyer_regular}, which says that the revenue maximizing randomized take-it-or-leave-it price that we derived in Theorem~\ref{thm:rl-random} is the optimal mechanism for a single buyer, assuming the regularity condition.  Due to space constraints, we will only give a proof sketch that demonstrates our duality framework\evn{, and we} defer the full proof to Appendix~\ref{sec:app_single_buyer_duality}.

\begin{theorem}\label{thm:single_buyer_regular}
	Consider a single risk-loving buyer with exponential utility whose value is drawn from a regular distribution $f$.  Then, the optimal mechanism is the revenue maximizing randomized take-it-or-leave-it price.
\end{theorem}

In order to prove this theorem, first note that the revenue of the optimal mechanism is upper-bounded by the Lagrangian dual program of the linear program \eqref{eq:lp_obj}.  For the single buyer case the dual program of \eqref{eq:lp_obj} can be simplified as
\begin{align}
\text{Minimize} \quad &\sum_{k \in [K]} \nu_k & \label{eq:mech_dual}
\end{align}
Subject to the following constraints:
\begin{align}
& f(v_k)z_j + \sum_{k'}\left(\lambda_{kk'}u(v_k-z_j)-\lambda_{k'k}u(v_{k'}-z_j)\right)+\mu_k u(v_k-z_j) \le \nu_k, &\forall k \in [K],j \in [M] \nonumber \\
& f(v_k)z_j + \sum_{k'}\left(\lambda_{kk'}-\lambda_{k'k}\right)u(-z_j) + \mu_ku(-z_j) \le \nu_k, &\forall k \in [K], j \in [M] \nonumber\\
& \mu_k \ge 0, \lambda_{kk'} \ge 0, &\forall k,k' \in [K] \nonumber,
\end{align}
where $\lambda_{kk'}$ corresponds to the BIC constraints  \eqref{eq:lp_cbic}, $\mu_k$ corresponds to the IR constraints  \eqref{eq:lp_cir}, and $\nu_k$ corresponds to the feasibility constraints  \eqref{eq:lp_cf}.  For the first two sets of the constraints in \evn{program} \eqref{eq:mech_dual}, we define
\begin{align*}
\Gamma_k(z;\lambda,\mu) &=f(v_k)z + \sum_{k'}\left(\lambda_{kk'}u(v_k-z)-\lambda_{k'k}u(v_{k'}-z)\right)+\mu_k u(v_k-z) \\
\Pi_k(z;\lambda,\mu) &=f(v_k)z + \sum_{k'}\left(\lambda_{kk'}-\lambda_{k'k}\right) u(-z)+\mu_k u(-z).
\end{align*}
In the dual program \eqref{eq:mech_dual}, the constraint $\Gamma_k(z;\lambda,\mu) \le \nu_k$ corresponds to the variable $y_{1,j}^1(v_k)$ in the primal \eqref{eq:lp_obj} and the constraint $\Pi_k(z;\lambda,\mu) \le \nu_k$ corresponds to the variable $y_{1,j}^0(v_k)$.  In fact, for any $k \in [K]$, we can show that among the sets of constraints $\{\Gamma_k(z_j;\lambda,\mu) \le \nu_k\}_{j \in [M]}$ and $\{\Pi_k(z_j;\lambda,\mu) \le \nu_k\}_{j \in [M]}$, only the following four can be binding: $\Gamma_k(0;\lambda,\mu)\le \nu_k$,  $\Gamma_k(z_M;\lambda,\mu)\le \nu_k$, $\Pi_k(0;\lambda,\mu)\le \nu_k$, and $\Pi_k(z_M;\lambda,\mu)\le \nu_k$.  This can be observed by the following lemma:
\begin{lemma} \label{lemma:exp_two_price1}
	In the dual program \eqref{eq:mech_dual}, both $\Gamma_k(z;\lambda,\mu)$ and $\Pi_k(z;\lambda,\mu)$ are either increasing or strongly convex in $z$ for $z \ge 0$.
\end{lemma}
The proof of this lemma can be found in Appendix~\ref{sec:app_single_buyer_duality}.
Now, we are ready to demonstrate how we can construct a set of dual variables that helps prove our main theorem in the following proof sketch.  Due to space constraints, the full calculation can be found in Appendix~\ref{sec:app_single_buyer_duality}.
\begin{proof}[Proof sketch of Theorem~\ref{thm:single_buyer_regular}]
	To prove the theorem, we construct a set of feasible dual variables that upper bound the value of the dual program \eqref{eq:mech_dual} by the revenue of the optimal take-it-or-leave-it randomized price.
	For any $k,k' \in [K]$, we set $\mu_k=0$ and 
	\begin{equation}
	\lambda_{kk'} = 
	\begin{cases}
	\sum_{\ell \ge k} f(v_{\ell}) z_M \frac{1}{u(v_k)-u(v_k-z_M)}, &\text{if }  k' = k-1\\
	0, &\text{otherwise}.
	\end{cases}
	\label{eq:single_lambda}
	\end{equation}
	For simplicity, we define $k^* = \min\{k: \phi_u(k)>0\}$. For any $k \in [K]$, we set
	\[
	\nu_k = \max \left\{ 0, f(v_k)\phi_u(k)z_M  \right\}.
	\]
	Under this choice of $\nu_k$, given the regularity condition, we can write the objective of the dual program \eqref{eq:mech_dual} as
	\[
	\sum_{k} \nu_k = \sum_{k: \phi_u(k)>0} f(v_k) \phi_u(k) z_M = \sum_{k \ge k^*} f(v_k) \phi_u(k) z_M.
	\]
	Then, by the definition of the virtual value function, we have
	\begin{align}
	\sum_{k} \nu_k &= \sum_{k \ge k^*} \left(\frac{ \sum_{k' \ge k}f(v_{k'}) u(v_k)}{u(v_k)-u(v_k-z_M)} -  \frac{\sum_{k' \ge {k+1}}f(v_{k'}) u(v_{k+1})}{u(v_{k+1})-u(v_{k+1}-z_M)} \right) z_M \nonumber \\
	&= \sum_{k \ge k^*}f(v_{k})z_M \cdot \frac{u(v_{k^*})}{u(v_{k^*})-u(v_{k^*}-z_M)},\label{eq:dual_obj_val}
	\end{align}
	where the last equality results from taking a telescopic sum.  We can find that \evn{the expression in the last line} \eqref{eq:dual_obj_val} is equal to the revenue of the optimal take-it-or-leave-it mechanism.
	
	It remains to verify that this choice of the dual variables is feasible for the dual program \eqref{eq:mech_dual}.  	
	Since we have argued that the sets of constraints $\{\Gamma_k(z_j;\lambda,\mu) \le \nu_k\}_{j \in [M]}$ and $\{\Pi_k(z;\lambda,\mu) \le \nu_k\}_{j \in [M]}$ can only be binding at $j=0$ and $j=M$, it suffices to check whether $\Gamma_k(0;\lambda,\nu) \le \nu_k$, $\Gamma_k(z_M;\lambda,\nu) \le \nu_k$, $\Pi_k(0;\lambda,\nu) \le \nu_k$, and $\Pi_k(z_M;\lambda,\nu) \le \nu_k$.  We defer the \evn{detailed} derivation to the full proof in Appendix~\ref{sec:app_single_buyer_duality}.  Here, we only show that by bringing our choice of the dual variables into these constraints, we have
	\begin{align*}
	\Gamma_k(0;\lambda,\mu) &= f(v_k) \phi_u(k)z_M , \quad
	&\Gamma_k(z_M;\lambda,\mu) &= f(v_k) \phi_u(k)z_M, \\
	\Pi_k(0;\lambda,\mu) &= 0,
	&\Pi_k(z_M;\lambda,\mu) &= (1-e^{-\alpha z_M}) f(v_k)\phi_u(k)z_M.
	\end{align*}
	By definition of $\nu_k$, we can find that this assignment of dual variables is feasible. Therefore, weak duality implies that the objective value of  \evn{program} \eqref{eq:lp_obj} is upper bounded by the revenue maximizing take-it-or-leave-it randomized price, which shows that this mechanism is optimal.
	
	In addition, we can find that given $k \in [K]$, as long as $\phi_u(k)>0$, the only binding dual constraints are $\Gamma_k(0;\lambda,\mu) \le \nu_k$ and $\Gamma_k(z_M;\lambda,\mu) \le \nu_k$. Therefore, by complementary slackness, in the optimal mechanism, the pricing scheme must be a randomization of $0$ and $z_M$, and must ask the buyer to pay only if she is given the item.  This coincides with the revenue-maximizing randomized take-it-or-leave-it price, which is our claimed optimal primal solution.
	
\end{proof}

\subsection{Optimal Mechanism beyond Regularity Condition}
\label{sec:single_ironing}
In the next step, we would like to extend our duality framework described in Theorem~\ref{thm:single_buyer_regular} to the case without a regularity condition.  Namely, we would like to prove that the revenue-maximizing take-it-or-leave-it randomized price is still optimal even though the virtual value function is not an increasing function, using the similar argument that we have made in Section~\ref{sec:sb_exp1}.  Formally, we would like to prove the following theorem:

\begin{theorem}\label{thm:mec_exputil}
	Consider a single risk-loving buyer with exponential utility.  The optimal mechanism is the revenue-maximizing randomized take-it-or-leave-it price.
\end{theorem}
The technique that we use to prove this theorem consists of two steps.  First, we construct an \emph{ironed virtual value} $\widetilde{\phi}_u$.  The ironed virtual value function is an increasing function constructed based on the original virtual value function, which is similar to the risk-neutral case as in Myerson's work.  After that, we can slightly modify the dual variables that we specified in Section~\ref{sec:sb_exp1} to match the ironed virtual value and then claim the optimality of the revenue-maximizing randomized take-it-or-leave-it price by strong duality.

Surprisingly, having a non-linear utility function does not complexify the ironing process.  We can directly apply Myerson's ironing for a risk-loving buyer by \emph{convexifying} the cumulative virtual value function.  Formally, we can define the ``ironed virtual value'' function in the following way:
\begin{definition}[Ironed virtual value for a single buyer]
	Given a virtual value function $\phi_u$.  Let $\{[a_1,b_1],[a_2,b_2],\dots,[a_m,b_m]\}$ denote the intervals that are not convex on the cumulative virtual value function $F_{\phi}(k) = \sum_{\ell \le k} \phi_u(\ell)$.  Then, the ironed virtual value function is defined as
	\[
	\widetilde{\phi}_{u}(k) = 
	\begin{cases}
	\frac{ \sum_{\ell = a_i}^{b_i} f(v_{\ell}) \phi_u(\ell) }{ \sum_{\ell = a_i}^{b_i} f(v_{\ell}) } , &\text{if }  k \in [a_i,b_i], \text{ for any } i \in [m]\\
	\phi_u(\ell), &\text{otherwise}.
	\end{cases}
	\]
\end{definition}
Although the ironing process is straightforward, how to modify the dual variables to match the ironed virtual value is more tricky.  Recall that under the choice of $\lambda$ as specified in \eqref{eq:single_lambda}, the left hand side of the first dual constraint is upper-bounded by the virtual value, i.e. $\Gamma_k(0;\lambda,\mu) = \Gamma_k(z_M;\lambda,\mu) = f(v_k)\phi_u(k)z_M$.  Motivated by \cite{cai2016duality}, we show that we can add \emph{loops} to $\lambda$ to alter the value of $\Gamma_k$.\footnote{In \cite{cai2016duality}, the dual variables can be interpreted as \emph{flows}.  However, in our setting, this interpretation no longer holds.  We need to handle the distortion caused by the non-linear utility function. }  More precisely, consider some $k' < k$.  If we add $Ae^{-\alpha v_{k}}$ to $\lambda_{k,k'}$ and add $Ae^{-\alpha v_{k'}}$ to $\lambda_{k',k}$ for some $A>0$, then $\Gamma_{k}(0;\lambda,\mu)$ is increased by $A[e^{-\alpha v_{k'}} - e^{-\alpha v_{k}}]$ and $\Gamma_{k'}(0;\lambda,\mu)$ is decreased by $A[e^{-\alpha v_{k'}} - e^{-\alpha v_{k}}]$. This distorted modification guarantees that the binding structure does not change, i.e. $\Gamma_k(0;\lambda,\mu)=\Gamma_k(z_M;\lambda,\mu)$ and $\Gamma_{k'}(0;\lambda,\mu)=\Gamma_{k'}(z_M;\lambda,\mu)$.  With this idea, we are able to apply Myerson's ironing by iterative adding loops to $\lambda$ in the way such that $\Gamma_k(0;\lambda,\mu) = f(v_k)\widetilde{\phi}_u(k) z_M$ holds after the modification.  After that, we are able to follow the steps in the proof of Theorem~\ref{thm:main_single_buyer} to prove the theorem.  Due to the space constraint, we defer the detailed proof of Theorem~\ref{thm:mec_exputil} to Appendix~\ref{sec:app_iron_single}.  
\section{Optimal Mechanism Design with Multiple Symmetric Buyers} \label{sec:multi-buyer}

In this section, we extend our analysis of optimal mechanism design to multiple symmetric buyers ($n\ge 2$).  Since the buyers are symmetric, their values come from the same distribution and they have the same utility function.
We show that the \emph{loser-pay auction} achieves the maximum revenue in the multiple-buyer case. In a loser-pay auction, the \evn{buyer with the highest value wins the item and only the buyers that do not obtain the item pay}. In addition,  similarly to the randomized take-it-or-leave-it pricing, all payments are made using a mixing of the minimum and the maximum price. This auction could be thought of as implementing an incentive-compatible version of the all-pay auction but adjusting it to achieve maximum discrepancy between the two outcomes. 

When characterizing the revenue maximizing mechanisms, our analysis is similar to the analysis in Section~\ref{sec:sb_exp1} in that it uses the virtual value formulation and the duality framework to upper bound the revenue obtained by the optimal mechanism. In what follows, we first give an example of the loser-pay auction and show how it improves the revenue compared to the second price auction, which is optimal for risk-neutral buyers.

\subsection{An Example}

\begin{example} \label{ex:two_buyers}
	Consider two buyers. Assume the private values of both buyers are distributed \evn{independently} according to the uniform distribution $U(\{0,1\})$, and $u(x)=e^{\alpha x}-1$.  Also assume $3(e^{\alpha x}-1) < 1- e^{-\alpha z_M}$.  Consider the following mechanism:
	\begin{enumerate}
		\item The item is allocated to the buyer who reports the higher value.  \evn{If both buyers} report $1$, the item is allocated uniformly at random.  \evn{If both buyers} report $0$, the item is not allocated to anyone.
		\item If a buyer reports $1$, she gets the item and does not pay anything.  However, if she reports $1$ and she does not get the item, then she pays $z_M$ with probability $\frac{3(e^{\alpha}-1)}{1-e^{-\alpha z_M}}$.
	\end{enumerate}  
	To verify this mechanism is BIC and IR, we check the utility curve if a buyer reports $1$ to the seller.  Consider buyer $1$. Given her true value is $t_1$, her expected utility is
	\[
	x_1(2) u(t_1) + (1-x_1(2)) \frac{3(e^{\alpha} -1)}{1-e^{-\alpha z_M}} u(-z_M) = \frac{3}{4} \left(e^{\alpha t_1}-1\right) - \frac{3}{4}(e^{\alpha}-1), 
	\]
	which is $0$ if $t_1=1$ and $-\frac{3}{4}(e^{\alpha}-1)$ if $t_1=0$.  This verifies that the mechanism is BIC and IR.  Next, we can find that the revenue of this mechanism is
	\[
	\text{Rev} = \sum_{i \in \{1,2\}} \Pr[t_i = 1] (1-x_i(2)) \frac{3(e^{\alpha}-1)}{1-e^{-\alpha z_M}} z_M = \frac{3}{4} \frac{e^{\alpha}-1}{1-e^{-\alpha z_M}} z_M.
	\]
\end{example}
From the above example, we make two observations.  First, we find that compared with the case of a single buyer, the revenue is increased by a factor of $\frac{3}{2}e^{\alpha}$.  Note that in the case of risk-neutral buyers, this factor is only $\frac{3}{2}$.  The additional factor of $e^{\alpha}$ comes from the second observation --- the buyer \emph{pays if she does not get the item}.  In the rest of this section, we show that these two properties hold in the optimal mechanism.

\subsection{The Loser-pay Auction}
\label{sec:multi_buyer_intro}
Recall that in the setting with risk-neutral buyers, the second-price auction with reserve price is optimal.  A natural \evn{question here} is, when the buyers are risk-loving, does the optimal mechanism take a similar form?  We show that, given the assumption that $z_M \gg \alpha$, i.e., the maximum allowed price is far greater than the level of risk-loving, the optimal mechanism corresponds to the revenue maximizing  loser-pay auction.  From Example~\ref{ex:two_buyers}, we already know that the optimal payment rule is to ask the buyer who loses to pay a randomized price.  The reserve price in our risk-loving setting, similarly to the risk-neutral setting, is going to be computed via a \emph{virtual value function}.

We have already defined a virtual value function for the convex utility function in Section~\ref{sec:sb_exp1}.  However, recall that the virtual value function is the marginal revenue in the quantile space.  According to Example~\ref{ex:two_buyers}, there is an additional $e^{\alpha}$ factor in the revenue of the multi-buyer setting, hence now we need a different definition of the virtual value function than in the single buyer case.    \evn{Specifically}, in the multi-buyer setting, we need to consider the following new virtual value function that takes this $e^{\alpha}$ factor into account:
\begin{definition}[Virtual value function for multiple buyers]
	\label{defn:virtual_value_multi_buyers}
	In the multi-buyer case, the virtual value function $\Phi_u: [K] \rightarrow \R$ with respect to an exponential utility function $u$ is defined as
	$$
	\Phi_u(k) = \frac{1}{f(v_k)} \left( \sum_{k' \ge k}f(v_{k'}) \cdot \frac{e^{\alpha v_k} u(v_k)}{u(v_k)-u(v_k-z_M)} - \sum_{k' \ge {k+1}}f(v_{k'}) \cdot \frac{e^{\alpha v_{k+1}} u(v_{k+1})}{u(v_{k+1})-u(v_{k+1}-z_M)} \right).
	$$
\end{definition}
In Section~\ref{sec:multi_duality}, we will show that this additional $e^{\alpha}$ factor comes from the \emph{competition} that only happens when there are multiple buyers.  Similar to the single buyer case, we call a distribution $f$ a \emph{regular distribution} if its corresponding virtual value function is a monotone increasing function.  If $f$ is not regular, then we can apply Myerson's ironing to the new virtual value function as what we have described in Section~\ref{sec:single_ironing}.  We let $\widetilde{\Phi}_u$ denote the ironed virtual value of $\Phi_u$.  The formal definition of $\widetilde{\Phi}_u$ can be found in Appendix~\ref{sec:app_iron_multiple}.

Now, we are ready to describe the loser-pay auction --- the mechanism that we claim to be optimal when the buyers are risk-loving:
\begin{definition}[The loser-pay auction]
	A loser-pay auction is a direct mechanism with the following allocation and payment rule:
	\begin{enumerate}
		\item Suppose each buyer $i$ bids $v_{k_i}$. Then, the auctioneer allocates the item to buyer $i$ if $\widetilde{\Phi}_u(k_i)>\widetilde{\Phi}_u(k_{i'})$ for every other buyer $i' \not= i$ provided $\widetilde{\Phi}_u(k_i)>0$.  
		\item Suppose each buyer $i$ bids $v_{k_i}$.  If buyer $i$ submits the bid with the largest ironed virtual value and ties with $n_t-1$ other buyers, then she gets the item with probability $1/n_t$ provided $\widetilde{\Phi}_u(k_i)>0$.
		\item If buyer $i$ bids $v_k$ and gets the item, she pays nothing.
		\item If buyer $i$ bids $v_k$ with $\widetilde{\Phi}_u(k) \le 0$, she pays nothing, i.e. $q_i(k)=0$.
		\item If buyer $i$ bids $v_k$ with $\widetilde{\Phi}_u(k) >0$ and does not get the item, then she pays $z_M$ with probability
		\begin{align*}
		q_i(k) = \frac{1}{1-x_i(k)}\sum_{k' = k^*}^{k} \frac{[x_i(k')-x_i(k'-1)]u(v_{k'})}{-u(-z_M)},
		\end{align*}
		where $k^* = \min\{k \in [K]: \Phi_u(k)>0\}$ is the index of the reserve price.
	\end{enumerate}
\end{definition}
For simplicity, in this section, we use $\vv{k}=(k_1,\dots,k_n)$ to denote the indices of bids from all buyers.  Also, we use $f(\vv{k}) = \prod_{i \in [n]} f(v_{k_i})$ to denote the probability that $\tv=(v_{k_1},v_{k_2},\dots,v_{k_n})$.
If buyer $i$ submits her bid $v_k$ that is no less than the reserve price $v_{k^*}$, then the interim probability that she gets the item is 
$$ x_i(k)=\sum_{\vv{k}_{-i} \in [K]^{n-1}} f(\vv{k}_{-i})\left[ \frac{\mathbbm{1}\left\{ \widetilde{\Phi}(v_{k_i}) \ge \widetilde{\Phi}(v_{k_{i'}}), \forall i' \not= i \right\}}{\sum_{i' \in [n]} \mathbbm{1}\left\{ \widetilde{\Phi}(v_{k_i}) = \widetilde{\Phi}(v_{k_{i'}}) \right\}} \right], $$
where $\mathbbm{1}\{A\}$ is the indicator function that equals  $1$ if event $A$ is true, and $0$ otherwise.  In order to guarantee that the payment rule of the loser-pay auction is feasible, we need to make the following assumption, which plays an important role in guaranteeing $q_i(k) <1$:

\noindent
(A1) $z_M \gg \alpha$ so that for each $v \in \setbvals$, it holds that $\frac{1-\frac{1}{n}f(v)}{\frac{1}{n}f(v)} \cdot (e^{\alpha v} - 1) < 1- e^{-\alpha z_M}$.

We show that the loser-pay auction is feasible, individually rational, and Bayesian incentive compatible under Assumption~(A1).  Due to space constraints, we defer the proof of the following theorem to Appendix~\ref{sec:app_multi_buyer_ppt}.

\begin{theorem} \label{lem:multibuyer_mech_feasible}
Consider $n \ge 2$ buyers with exponential utility.  With Assumption~(A1), the loser-pay auction is feasible, individually rational, and Bayesian incentive compatible.
\end{theorem}

\subsection{Duality Theory for Optimal Mechanism Design}
\label{sec:multi_duality}

We now show that the loser-pay auction is \evn{an} optimal mechanism.  Formally, we will show the following theorem, which is our main result when we have $n \ge 2$ risk-loving buyers:
\begin{theorem}\label{thm:two_buyer_regular}
	Consider $n \ge 2$ buyers with exponential utility.  With Assumption~(A1), the loser-pay auction is the optimal mechanism.
\end{theorem}

The core idea behind the proof of Theorem~\ref{thm:two_buyer_regular} is similar to Theorem~\ref{thm:single_buyer_regular}.  First, we write down the dual program of the linear program \eqref{eq:lp_obj}.  Then, we construct a set of feasible dual variables and show that under \evn{this} assignment, the dual objective is equal to the revenue of the loser-pay auction.  Then, by strong duality of linear programming, we conclude that the revenue maximizing loser-pay auction is optimal.  When we have $n \ge 2$ symmetric buyers, the dual program of the linear program \eqref{eq:lp_obj} can be derived to be:
\Ger{\begin{align}
\text{Minimize} \quad &\sum_{\vv{k} \in [K]^n} \left(\sum_{i \in [n]} \nu_{i,\vv{k}}  + \gamma_{\vv{k}} \right) & \label{eq:mech_dual_multibuyer}
\end{align}}
Subject to the following constraints:
\begin{align}
& f(\vv{k}_{-i})\Gamma_{i,k_i}(z_j;\lambda,\mu) \le \nu_{i,\vv{k}} + \gamma_{\vv{k}},  &\forall i \in [n], \forall \vv{k} \in [K]^n ,j \in [M] \nonumber \\
& f(\vv{k}_{-i})\Pi_{i,k_i}(z_j;\lambda,\mu) \le \nu_{i,\vv{k}}, &\forall i \in [n], \forall \vv{k} \in [K]^n, j \in [M] \nonumber \\
& \mu_{\vv{k}} \ge 0, \lambda_{i,k_i,k_i'} \ge 0, \gamma_{\vv{k}} \ge 0 &\forall i \in [n], \forall \vv{k} \in [K]^n, \nonumber
\end{align}
where
\begin{align*}
\Gamma_{i,k}(z;\lambda,\mu) &=f(v_k)z + \sum_{k'}\left(\lambda_{i,k,k'}u(v_k-z)-\lambda_{i,k',k}u(v_{k'}-z)\right)+\mu_{i,k} u(v_k-z) \\
\Pi_{i,k}(z;\lambda,\mu) &=f(v_k)z + \sum_{k'}\left(\lambda_{i,k,k'}-\lambda_{i,k',k}\right) u(-z)+\mu_{i,k} u(-z).
\end{align*}
In the dual program \eqref{eq:mech_dual_multibuyer}, the dual variable $\lambda_{i,k,k'}$ corresponds to the BIC constraint \eqref{eq:lp_cbic} for buyer $i$, and the dual variable $\mu_{i,k}$ corresponds to the IR constraints \eqref{eq:lp_cir} for buyer $i$.  Differently from the single buyer case in program \eqref{eq:mech_dual}, for the feasibility constraints \eqref{eq:lp_cf}, we have the dual variable $\nu_{i,\vv{k}}$ associated with each buyer, and $\gamma_{\vv{k}}$ associated with the single-item constraint.  Also, the constraint $f(\vv{k}_{-i})\Gamma_{i,k_i}(z_j;\lambda,\mu) \le \nu_{i,\vv{k}} + \gamma_{\vv{k}}$ corresponds to the primal variable $y_{i,j}^1(v_{k_1},\dots, v_{k_n})$ and the constraint $f(\vv{k}_{-i})\Pi_{i,k_i}(z_j;\lambda,\mu) \le \nu_{i,\vv{k}}$ corresponds to the primal variable $y_{i,j}^0(v_{k_1},\dots,v_{k_n})$.

Similarly to Lemma~\ref{lemma:exp_two_price1}, we show in Lemma~\ref{lemma:two_buyer_exp_two_price1} in the appendix that both $\Gamma_{i,k}(z;\lambda,\mu)$ and $\Pi_{i,k}(z;\lambda,\mu)$ are either increasing or strongly convex in $z$.  This means that the first two sets of constraints in \eqref{eq:mech_dual_multibuyer} can only be binding at either end points $j=1$ or $j=M$.

If $f$ is a regular distribution, then we can use the following assignment to prove the strong duality: For each buyer $i \in [n]$ and any $k,k' \in [K]$, we set $\mu_{i,k}=0$ and 
\begin{equation}
\lambda_{i,k,k'} = 
\begin{cases}
\sum_{\ell \ge k} f(v_{\ell}) z_M \frac{e^{\alpha v_k}}{u(v_k)-u(v_k-z_M)}, &\text{if }  k' = k-1\\
0, &\text{otherwise.}
\end{cases}
\label{eq:multi_lambda}
\end{equation}
For each buyer $i \in [n]$ and any $\vv{k} \in [K]^n$, we also set $\nu_{i,\vv{k}}=0$ and $\gamma_{\vv{k}} = \max_{i \in [n]} \left\{ 0, f(\vv{k})\Phi_u(k_i)z_M  \right\}$.  Under this assignment, we have $f(\vv{k}_{-i})\Pi_{i,k_i}(0; \lambda,\mu) = f(\vv{k}_{-i})\Pi_{i,k_i}(z_M; \lambda,\mu) = 0$ and
\begin{align*}
&f(\vv{k}_{-i})\Gamma_{i,k_i}(0; \lambda,\mu) = f(\vv{k})\Phi_u(k_i)z_M,
\quad
&f(\vv{k}_{-i})\Gamma_{i,k_i}(z_M; \lambda,\mu) = f(\vv{k})\Phi_u(k_i)z_M e^{-\alpha z_M}.
\end{align*}
Differently from the assignment \eqref{eq:single_lambda} that we made in the single buyer case, the binding structure changes here because of the additional $\gamma_{\vv{k}}$ in the first dual constraint from the single item constraint.  For $k>k^*$, the dual constraints no longer bind at $\Gamma_{i,k}(0;\lambda,\mu)$ and $\Gamma_{i,k}(z_M;\lambda,\mu)$ but at $\Gamma_{i,k}(0;\lambda,\mu)$ and $\Pi_{i,k}(z_M;\lambda,\mu)$.  The additional $e^{\alpha}$ factor in \eqref{eq:multi_lambda} ensures this new binding pattern.

If $f$ is not a regular distribution, we can add \emph{loops} to $\lambda$ as we have done in Section~\ref{sec:single_ironing} for the single buyer case.  Due to the space constraint, we defer the full discussion with the regularity condition to Appendix~\ref{sec:app_multi_buyer_opt} and the discussion without regularity condition to Appendix~\ref{sec:app_iron_multiple}.

\section{Conclusion} \label{sec:conclusion}
In this paper we studied the revenue-maxizining mechanism for  the special case of risk-loving agents with exponential utility functions. We demonstrate that  for both a single and multiple symmetric buyers the optimal auction is simple. A natural question is whether or not the same results extend beyond the expoential utility function. Unfortunately, it can be shown that optimal auction for the case of a single agent with quadratic utility function is more complicated than the simple randomized take-it-or-leave it offer and requires us to utilize at least one more menu option as shown in Appendix~\ref{quadratic}. 

\bibliographystyle{plainnat}
\bibliography{refs}

\begin{thebibliography}{17}
\providecommand{\natexlab}[1]{#1}
\providecommand{\url}[1]{\texttt{#1}}
\expandafter\ifx\csname urlstyle\endcsname\relax
  \providecommand{\doi}[1]{doi: #1}\else
  \providecommand{\doi}{doi: \begingroup \urlstyle{rm}\Url}\fi

\bibitem[Cai et~al.(2016)Cai, Devanur, and Weinberg]{cai2016duality}
Yang Cai, Nikhil~R. Devanur, and S.~Matthew Weinberg.
\newblock A duality based unified approach to bayesian mechanism design.
\newblock In \emph{Proceedings of the forty-eighth annual ACM symposium on
  Theory of Computing}, pages 926--939. ACM, 2016.

\bibitem[Chawla et~al.(2018)Chawla, Goldner, Miller, and
  Pountourakis]{chawla2017aversion}
Shuchi Chawla, Kira Goldner, J~Benjamin Miller, and Emmanouil Pountourakis.
\newblock Revenue maximization with an uncertainty-averse buyer.
\newblock In \emph{Proceedings of the Twenty-Ninth Annual ACM-SIAM Symposium on
  Discrete Algorithms}, pages 2050--2068. SIAM, 2018.

\bibitem[Daskalakis et~al.(2013)Daskalakis, Deckelbaum, and
  Tzamos]{daskalakis2013mechanism}
Constantinos Daskalakis, Alan Deckelbaum, and Christos Tzamos.
\newblock Mechanism design via optimal transport.
\newblock In \emph{Proceedings of the fourteenth ACM conference on Electronic
  commerce}, pages 269--286. ACM, 2013.

\bibitem[Daskalakis et~al.(2017)Daskalakis, Deckelbaum, and
  Tzamos]{daskalakis2017strong}
Constantinos Daskalakis, Alan Deckelbaum, and Christos Tzamos.
\newblock Strong duality for a multiple-good monopolist.
\newblock \emph{Econometrica}, 85\penalty0 (3):\penalty0 735--767, 2017.

\bibitem[Dughmi and Peres(2012)]{dughmi2012mechanisms}
Shaddin Dughmi and Yuval Peres.
\newblock Mechanisms for risk averse agents, without loss.
\newblock \emph{arXiv preprint arXiv:1206.2957}, 2012.

\bibitem[Fu et~al.(2013)Fu, Hartline, and Hoy]{fu2013prior}
Hu~Fu, Jason Hartline, and Darrell Hoy.
\newblock Prior-independent auctions for risk-averse agents.
\newblock In \emph{Proceedings of the fourteenth ACM conference on Electronic
  commerce}, pages 471--488. ACM, 2013.

\bibitem[Giannakopoulos and Koutsoupias(2014)]{giannakopoulos2014duality}
Yiannis Giannakopoulos and Elias Koutsoupias.
\newblock Duality and optimality of auctions for uniform distributions.
\newblock In \emph{Proceedings of the fifteenth ACM conference on Economics and
  computation}, pages 259--276. ACM, 2014.

\bibitem[Giannakopoulos and Koutsoupias(2015)]{giannakopoulos2015selling}
Yiannis Giannakopoulos and Elias Koutsoupias.
\newblock Selling two goods optimally.
\newblock In \emph{International Colloquium on Automata, Languages, and
  Programming}, pages 650--662. Springer, 2015.

\bibitem[Hinnosaar(2017)]{hinnosaar2016impossibility}
Toomas Hinnosaar.
\newblock On the impossibility of protecting risk-takers.
\newblock \emph{The Economic Journal}, 2017.
\newblock ISSN 1468-0297.
\newblock \doi{10.1111/ecoj.12446}.
\newblock URL \url{http://dx.doi.org/10.1111/ecoj.12446}.

\bibitem[Li et~al.(2015)Li, Xia, Geng, Ming, Shakkottai, Subramanian, and
  Xie]{li2015energy}
Jian Li, Bainan Xia, Xinbo Geng, Hao Ming, Srinivas Shakkottai, Vijay
  Subramanian, and Le~Xie.
\newblock Energy coupon: A mean field game perspective on demand response in
  smart grids.
\newblock \emph{ACM SIGMETRICS Performance Evaluation Review}, 43\penalty0
  (1):\penalty0 455--456, 2015.

\bibitem[Lu(2015)]{lu2015framework}
Fangping Lu.
\newblock \emph{Framework for a lottery-based incentive scheme and its
  influence on commuting behaviors: an MIT case study}.
\newblock PhD thesis, Massachusetts Institute of Technology, 2015.

\bibitem[Maskin and Riley(1984)]{maskin1984optimal}
Eric Maskin and John Riley.
\newblock Optimal auctions with risk averse buyers.
\newblock \emph{Econometrica: Journal of the Econometric Society}, pages
  1473--1518, 1984.

\bibitem[Matthews(1983)]{matthews1983selling}
Steven~A Matthews.
\newblock Selling to risk averse buyers with unobservable tastes.
\newblock \emph{Journal of Economic Theory}, 30\penalty0 (2):\penalty0
  370--400, 1983.

\bibitem[Merugu et~al.(2009)Merugu, Prabhakar, and Rama]{balaji2}
Deepak Merugu, Balaji~S Prabhakar, and N.~S. Rama.
\newblock An incentive mechanism for decongesting the roads: A pilot program in
  bangalore.
\newblock In \emph{Proc. of ACM NetEcon Workshop}. ACM, 2009.

\bibitem[Myerson(1981)]{myerson1981optimal}
Roger~B Myerson.
\newblock Optimal auction design.
\newblock \emph{Mathematics of operations research}, 6\penalty0 (1):\penalty0
  58--73, 1981.

\bibitem[Pluntke and Prabhakar(2013)]{balaji1}
Christopher Pluntke and Balaji Prabhakar.
\newblock Insinc: A platform for managing peak demand in public transit.
\newblock \emph{JOURNEYS, Land Transport Authority Academy of Singapore}, pages
  31--39, 2013.

\bibitem[Von~Neumann and Morgenstern(1945)]{von1945theory}
John Von~Neumann and Oskar Morgenstern.
\newblock Theory of games and economic behavior.
\newblock \emph{Bull. Amer. Math. Soc}, 51\penalty0 (7):\penalty0 498--504,
  1945.

\end{thebibliography}
\newpage

\section{Appendix}
\appendix
\section{Missing Proofs from Optimal Mechanism Design for a Single Buyer} \label{sec:app_single_buyer}

\subsection{Optimal Take-it-or-leave-it Randomized Price}
\label{sec:app_single_buyer1}
\begin{proof}[Proof of Theorem~\ref{thm:rl-random}]
	Consider the menu option $(1,\Pv)$ with $\E[\Pv]=p$.  Suppose $\Pv$ is not randomized only at $0$ and $z_M$. We will show that we can flip one more coin to construct a new randomized payment $\Pv'$ with only $0$ and $z_M$  and obtain potentially more revenue.  Conditioned on $\Pv$, we set $\Pv'=z_M$ with probability $\Pv/z_M$, and set $\Pv'=0$ otherwise.  Then, we can find that given the buyer's value is $v$, the expected utility of choosing $(1,\Pv')$ is
	\begin{align*} 
	\E[u(v-\Pv')] &= \E[(\Pv/z_M)u(v-z_M) + (1-\Pv /z_M)u(v)] \\
	&\ge \E[u((\Pv/z_M)(v-z_M) + (1-\Pv/z_M)v)] = \E[u(v-\Pv)] 
	\end{align*}
	where the inequality comes from Jensen's inequality.  This means that for any value, the expected utility curve of $(1,\Pv')$ is no less than $(1,\Pv)$.  Let $t_0$ denote the intersection of $\E[u(x-\Pv)]$ and the $x$-axis, i.e. let $t_0$ be such that
	$$ \E[u(t_0-\Pv)] = 0. $$
	Similarly, let $t_0'$ denote the intersection of $\E[u(x-\Pv')]$ and the $x$-axis, i.e.
	$$ \E[u(t_0'-\Pv')] = 0. $$
	We can find that since $\E[u(x-\Pv')]$ is above $\E[u(x-\Pv)]$, and both functions are monotone increasing, we have $t_0' \le t_0$. This implies that the revenue using a randomized take-it-or-leave-it  price can be achieved by randomizing only between $0$ and $z_M$. 
	
	As a result, given the expectation of a randomized payment, we can without loss of generality consider its implementation that randomizes over the extremes as follows: with probability $p/z_M$ charge $z_M$ and with the remaining probability charge $0$. 
	To find the probability that such a menu is accepted, we solve for the threshold value $t'$ that is indifferent between accepting and rejecting:
	
	$$ (p/z_M) u(t'-z_M) + (1-p/z_M) u(t') = 0. $$
	
	Solving for $p$, we can find that $p = \frac{z_M u(t')}{u(t') - u(t'-z_M)}$.
	Since the probability of a buyer choosing this menu option is $\Pr[t\ge t']$, we can conclude that the revenue is $\frac{z_M\Pr[t\ge t']u(t')}{u(t')-u(t'-z_M)}$.
\end{proof} 

\subsection{Duality Theory for Optimal Mechanisms}
\label{sec:app_single_buyer_duality}
First, we show that in the dual program \eqref{eq:mech_dual}, for any $k \in [K]$, among the sets of constraints $\{\Gamma_k(z_j;\lambda,\mu) \le \nu_k\}_{j \in [M]}$ and $\{\Pi_k(z_j;\lambda,\mu) \le \nu_k\}_{j \in [M]}$, only the following four can be binding: 
\begin{enumerate}
	\item $\Gamma_k(0;\lambda,\mu)\le \nu_k$
	\item $\Gamma_k(z_M;\lambda,\mu)\le \nu_k$
	\item $\Pi_k(0;\lambda,\mu)\le \nu_k$
	\item $\Pi_k(z_M;\lambda,\mu)\le \nu_k$
\end{enumerate}
This can be observed by noticing that both $\Gamma_k$ and $\Pi_k$ are either increasing or strongly convex in $z$ for $z \ge 0$, which we prove in the following lemma (restatement of Lemma~\ref{lemma:exp_two_price1}):
\begin{lemma} 
	In the dual program \eqref{eq:mech_dual}, both $\Gamma_k(z;\lambda,\mu)$ and $\Pi_k(z;\lambda,\mu)$ are either increasing or strongly convex in $z$ for $z \ge 0$.
\end{lemma}
\begin{proof}
	To show this, we can simply rewrite $\Gamma_k(z;\lambda,\mu)$ in the following way
	\begin{align}
	\Gamma_k(z;\lambda,\mu) &=f(v_k)z + \sum_{k'}\left(\lambda_{kk'}\beta(e^{\alpha(v_k-z)}-1)-\lambda_{k'k}\beta(e^{\alpha(v_{k'}-z)}-1)\right)+\mu_k \beta(e^{\alpha(v_k-z)}-1) \nonumber \\ 
	&=f(v_k)z + \beta \left[ \sum_{k'}\left(\lambda_{kk'}e^{\alpha v_k} - \lambda_{k'k}e^{\alpha v_{k'}}\right)e^{-\alpha z} +  \mu_k e^{\alpha v_k} e^{-\alpha z} -  \sum_{k'} \left(\lambda_{kk'} - \lambda_{k'k}\right) - \mu_k \right] \nonumber \\
	&=f(v_k)z + \beta B_k e^{-\alpha z} - \beta A_k,
	\end{align}
	where
	\begin{align*}
	A_k &= \sum_{k'}(\lambda_{kk'}-\lambda_{k'k}) + \mu_k \\
	B_k &= \sum_{k'}(\lambda_{kk'}e^{\alpha v_k} - \lambda_{k'k}e^{\alpha v_{k'}}) +  \mu_k e^{\alpha v_k}.
	\end{align*}
	We can find that if $B_k>0$, then $\Gamma_k(z;\lambda,\mu)$ is strongly convex in $z$ whereas if $B_k \le 0$, then $\Gamma_k(z;\lambda,\mu)$ is monotone increasing in $z$.
	Similarly, for $\Pi_k(z;\lambda,\mu)$, we have
	\begin{align}
	\Pi_k(z;\lambda,\mu) &=f(v_k)z + \sum_{k'}\left(\lambda_{kk'}-\lambda_{k'k}\right) \beta(e^{-\alpha z}-1)+\mu_k \beta(e^{-\alpha z}-1) \nonumber \\ 
	&=f(v_k)z + \beta A_k e^{-\alpha z} - \beta A_k,
	\end{align}
	Also, we can find that if $A_k>0$, then $\Pi_k(z;\lambda,\mu)$ is strongly convex in $z$ whereas if $A_k \le 0$, then $\Pi_k(z;\lambda,\mu)$ is monotone increasing in $z$.
\end{proof}
Now, we are ready to prove the main theorem.
\begin{proof}[Proof of Theorem~\ref{thm:single_buyer_regular}]
	To prove the theorem, we construct a set of feasible dual variables that upper bound the value of the dual program \eqref{eq:mech_dual} by the revenue of the optimal take-it-or-leave-it randomized price.
	
	For any $k,k' \in [K]$, we set $\mu_k=0$ and 
	\[
	\lambda_{kk'} = 
	\begin{cases}
	\sum_{\ell \ge k} f(v_{\ell}) z_M \frac{1}{u(v_k)-u(v_k-z_M)}, &\text{if }  k' = k-1\\
	0, &\text{otherwise}.
	\end{cases}
	\]
	For simplicity, we define $k^* = \min\{k: \phi_u(k)>0\}$. For any $k \in [K]$, we set
	\[
	\nu_k = \max \left\{ 0, f(v_k)\phi_u(k)z_M  \right\}.
	\]
	Note that under this choice of $\nu_k$, the objective of the dual program \eqref{eq:mech_dual} becomes
	\begin{align}
	\sum_{k} \nu_k &= \sum_{k: \phi_u(k)>0} f(v_k) \phi_u(k) z_M \nonumber \\
	&= \sum_{k \ge k^*} f(v_k) \phi_u(k) z_M \nonumber \\
	&= \sum_{k \ge k^*} \left(\frac{ \sum_{k' \ge k}f(v_{k'}) u(v_k)}{u(v_k)-u(v_k-z_M)} -  \frac{\sum_{k' \ge {k+1}}f(v_{k'}) u(v_{k+1})}{u(v_{k+1})-u(v_{k+1}-z_M)} \right) z_M \nonumber \\
	&= \sum_{k \ge k^*}f(v_{k})z_M \cdot \frac{u(v_{k^*})}{u(v_{k^*})-u(v_{k^*}-z_M)}, \label{eq:dual_obj_val2}
	\end{align}
	where the second equality follows from the regularity condition, and the last equality results from taking a telescopic sum.  We can find that \eqref{eq:dual_obj_val2} is equal to the revenue of the optimal take-it-or-leave-it mechanism.
	
	It remains to verify that this choice of the dual variables is feasible for the dual program \eqref{eq:mech_dual}.  By Lemma~\ref{lemma:exp_two_price1}, it suffices to check whether $\Gamma_k(0;\lambda,\mu) \le \nu_k$, $\Gamma_k(z_M;\lambda,\mu) \le \nu_k$, and $\Pi_k(z_M;\lambda,\mu) \le \nu_k$.
	First, we check:
	\begin{align*}
	\Gamma_k(0; \lambda,\mu) &= \lambda_{k,k-1} u(v_k) - \lambda_{k+1,k} u(v_{k+1})\\
	&= \sum_{\ell \ge k} f(v_{\ell}) z_M \frac{u(v_k)}{u(v_k)-u(v_k-z_M)} - \sum_{\ell \ge k+1} f(v_{\ell}) z_M \frac{u(v_{k+1})}{u(v_{k+1})-u(v_{k+1}-z_M)}\\
	&= f(v_k)\phi_u(k)z_M \le \nu_k
	\end{align*}
	and 
	\begin{align*}
	\Gamma_k(z_M; \lambda,\mu) &= f(v_k)z_M + \lambda_{k,k-1} u(v_k-z_M) - \lambda_{k+1,k} u(v_{k+1}-z_M)\\
	&=  f(v_k)z_M + \frac{\sum_{\ell \ge k} f(v_{\ell}) z_M u(v_k-z_M)}{u(v_k)-u(v_k-z_M)} -  \frac{\sum_{\ell \ge k+1} f(v_{\ell}) z_M u(v_{k+1}-z_M)}{u(v_{k+1})-u(v_{k+1}-z_M)}\\
	&=  \sum_{\ell \ge k} f(v_k)z_M \frac{u(v_k)-u(v_k-z_M)}{u(v_k)-u(v_k-z_M)} - \sum_{\ell \ge k+1} f(v_k)z_M \frac{u(v_{k+1})-u(v_{k+1}-z_M)}{u(v_{k+1})-u(v_{k+1}-z_M)} \\
	&\quad+ \sum_{\ell \ge k} f(v_{\ell}) z_M \frac{u(v_k-z_M)}{u(v_k)-u(v_k-z_M)} -  \sum_{\ell \ge k+1} f(v_{\ell}) z_M \frac{u(v_{k+1}-z_M)}{u(v_{k+1})-u(v_{k+1}-z_M)}\\
	&= \sum_{\ell \ge k} f(v_{\ell}) z_M \frac{u(v_k)}{u(v_k)-u(v_k-z_M)} - \sum_{\ell \ge k+1} f(v_{\ell}) z_M \frac{u(v_{k+1})}{u(v_{k+1})-u(v_{k+1}-z_M)}\\
	&= f(v_k)\phi_u(k)z_M \le \nu_k.
	\end{align*}
	For $\Pi_k(z_M;\lambda,\mu)$, we note that
	\begin{align*}
	f(v_k)\phi_u(k)z_M &= \sum_{\ell \ge k} f(v_{\ell}) z_M \frac{u(v_k)}{u(v_k)-u(v_k-z_M)} - \sum_{\ell \ge k+1} f(v_{\ell}) z_M \frac{u(v_{k+1})}{u(v_{k+1})-u(v_{k+1}-z_M)}\\
	&= \sum_{\ell \ge k} f(v_{\ell}) z_M \frac{1-e^{-\alpha v_k}}{1-e^{-\alpha z_M}} - \sum_{\ell \ge k+1} f(v_{\ell}) z_M \frac{1-e^{-\alpha v_{k+1}}}{1-e^{-\alpha z_M}}\\
	&= \frac{f(v_k)}{1-e^{-\alpha z_M}} + \sum_{\ell \ge k+1} f(v_{\ell}) z_M \frac{e^{-\alpha v_{k+1}}}{1-e^{-\alpha z_M}} - \sum_{\ell \ge k} f(v_{\ell}) z_M \frac{e^{-\alpha v_k}}{1-e^{-\alpha z_M}}.
	\end{align*}
	Then,
	\begin{align*}
	\Pi_k(z_M;\lambda,\mu) &= f(v_k)z_M + \lambda_{k,k-1} u(-z_M) - \lambda_{k+1,k} u(-z_M) \\
	&= f(v_k)z_M + \sum_{\ell \ge k} f(v_{\ell}) z_M \frac{u(-z_M)}{u(v_k)-u(v_k-z_M)} - \sum_{\ell \ge k+1} f(v_{\ell}) z_M \frac{u(-z_M)}{u(v_{k+1})-u(v_{k+1}-z_M)}\\
	&= f(v_k)z_M - \sum_{\ell \ge k} f(v_{\ell}) z_M e^{-\alpha v_k} + \sum_{\ell \ge k+1} f(v_{\ell}) z_M e^{-\alpha v_{k+1}}\\
	&=(1-e^{-\alpha z_M}) f(v_k)\phi_u(k)z_M \le \nu_k.
	\end{align*}
	Therefore, the choice of the dual variables is feasible.  Weak duality implies that the objective value of  \eqref{eq:lp_obj} is upper bounded by the revenue maximizing randomized take-it-or-leave-it  price, which shows that the revenue maximizing randomized take-it-or-leave-it price is optimal.
\end{proof}

\subsection{Optimal Mechanism beyond Regularity Condition}
\label{sec:app_iron_single}
As we have mentioned in Section~\ref{sec:single_ironing}, the main intuition behind the proof of Theorem~\ref{thm:mec_exputil} is that we can construct an \emph{ironed virtual value} $\widetilde{\phi}_u$.  Then, we can slightly modify the dual variables that we specified in Section~\ref{sec:sb_exp1} to match the ironed virtual value.  More specifically, consider the dual program \eqref{eq:mech_dual} and suppose we have a new set of dual variables $(\widetilde{\lambda},\widetilde{\mu},\widetilde{\nu})$.  We would like to show that 
\begin{enumerate}
	\item $\widetilde{\nu}_k = \max\{0,f(v_k)\widetilde{\phi}_u z_M\}$, which relates the dual variables with the ironed virtual value.
	\item $\sum_k \widetilde{\nu}_k = \max_{k \in [K]} \left[ \sum_{\ell \ge k} f(v_{\ell}) z_M \frac{u(v_{k})}{u(v_{k})-u(v_{k}-z_M)} \right]$, which shows the strong duality with the revenue-maximizing take-it-or-leave-it randomized price.
\end{enumerate}

The ironing argument that we are going to make in this section is motivated by \cite{cai2016duality}.  However, since in the risk-loving setting, we can no longer interpret the dual variables as ``flows'' as \cite{cai2016duality} have done in the risk-neutral setting, the challenge here is to identify how to carefully take care of the distortion caused by a non-linear utility function when adding a \emph{loop} to the dual variables. 

Before we give a formal proof of Theorem~\ref{thm:mec_exputil}, let us give some intuition that leads to the ironing process.  Consider the dual variables $\lambda'$.  Suppose $\lambda'$ coincides with $\lambda$ except at $k=k_0$ and $k=k_1$ for some $k_1<k_0$:
\begin{align}
\lambda_{k_0,k_1}' &= \lambda_{k_0,k_1} + 
\sum_{\ell \ge k_0} A z_M \frac{1}{u(v_{k_0})-u(v_{k_0}-z_M)} \label{eq:loop1}\\
\lambda_{k_1,k_0}' &= \lambda_{k_1,k_0} + \sum_{\ell \ge k_0} A z_M \frac{1}{u(v_{k_1})-u(v_{k_1}-z_M)} \label{eq:loop2}
\end{align}
for some $A > 0$.  Comparing this with the definition of $\lambda$ in \eqref{eq:single_lambda}, here we are introducing an additional \emph{loop} of amount $A$ between $k_0$ and $k_1$.  Note that this particular way of adding the loop ensures $\Gamma_{k}(0;\lambda',\mu)=\Gamma_{k}(z_M;\lambda',\mu)$ and $\Pi_k(z_M;\lambda',\mu) = (1-e^{-\alpha z_M})\Gamma_k(0;\lambda',\mu)$ for any $k \in [K]$.  Then, we can find that at $k=k_0$, we have
\begin{align*}
\Gamma_{k_0}(0; \lambda',\mu) &= \Gamma_{k_0}(0;\lambda, \mu) + \sum_{\ell \ge k_0} A z_M \left[ \frac{u(v_{k_0})}{u(v_{k_0})-u(v_{k_0}-z_M)} - \frac{u(v_{k_1})}{u(v_{k_1})-u(v_{k_1}-z_M)} \right] \\
&= \Gamma_{k_0}(0;\lambda, \mu) + (K - k_0 + 1) A z_M \left[ \frac{e^{-\alpha v_{k_1}} - e^{-\alpha v_{k_0}}}{1-e^{-\alpha z_M}} \right]\\
&> \Gamma_{k_0}(0;\lambda, \mu). \\
\end{align*}
Similarly, at $k_1$, we have
\begin{align*}
\Gamma_{k_1}(0; \lambda',\mu) &= \Gamma_{k_1}(0;\lambda, \mu) + \sum_{\ell \ge k_0} A z_M \left[- \frac{u(v_{k_0})}{u(v_{k_0})-u(v_{k_0}-z_M)} + \frac{u(v_{k_1})}{u(v_{k_1})-u(v_{k_1}-z_M)} \right] \\
&= \Gamma_{k_1}(0;\lambda, \mu) + (K - k_0 + 1) A z_M \left[ \frac{e^{-\alpha v_{k_0}} - e^{-\alpha v_{k_1}}}{1-e^{-\alpha z_M}} \right] \\
&< \Gamma_{k_1}(0;\lambda, \mu). \\
\end{align*}
We can find that after adding the loop, $\Gamma_{k_0}$ is increased by $(K - k_0 + 1) A z_M \left[ \frac{e^{-\alpha v_{k_1}} - e^{-\alpha v_{k_0}}}{1-e^{-\alpha z_M}} \right]$ and $\Gamma_{k_1}$ is decreased by $(K - k_0 + 1) A z_M \left[ \frac{e^{-\alpha v_{k_1}} - e^{-\alpha v_{k_0}}}{1-e^{-\alpha z_M}} \right]$.  Note that in the proof of Theorem~\ref{thm:single_buyer_regular}, we have shown that $\Gamma_{k}(0;\lambda,\nu) = f(v_k)\phi_u(k)z_M$ for any $k \in [K]$, which relates the dual constraint $\Gamma_k(0;\lambda,\nu)$ with the virtual value $\phi_u(k)$ at $k$.  With this idea, the process of \emph{adding a loop} in $\lambda$ is equivalent to modifying the virtual value function.  In addition, we can see that adding a loop does not change the average virtual value, i.e. $\Gamma_{k_0}(0;\lambda',\mu) + \Gamma_{k_1}(0;\lambda',\mu) = \Gamma_{k_0}(0;\lambda,\mu) + \Gamma_{k_1}(0;\lambda,\mu) = f(v_{k_0})\phi_u(k_0)z_M + f(v_{k_1})\phi_u(k_1)z_M$. 

We call the particular process of modifying the virtual value function into a monotone increasing function the \emph{ironing} process.  In the ironing process, we consider the convex hull of the \emph{revenue curve}, which is the cumulative function of the virtual value function, defined as
\[
F_{\phi}(k) = \sum_{\ell \le k} \phi_u(\ell).
\]
Then, we take the derivative of the revenue curve as the ironed virtual value function.  Formally, we can define the ironing process as follows:
\begin{definition}[Ironed virtual value for a single buyer]
	Given a virtual value function $\phi_u$.  Let $\{[a_1,b_1],[a_2,b_2],\dots,[a_m,b_m]\}$ denote the minimum intervals that are not convex on the revenue curve $F_{\phi}$.  The ironed virtual value function is defined as
	\[
	\widetilde{\phi}_{u}(k) = 
	\begin{cases}
	\frac{ \sum_{\ell = a_i}^{b_i} f(v_{\ell}) \phi_u(\ell) }{ \sum_{\ell = a_i}^{b_i} f(v_{\ell}) } , &\text{if }  k \in [a_i,b_i], \text{ for any } i \in [m]\\
	\phi_u(\ell), &\text{otherwise}.
	\end{cases}
	\]
\end{definition}

Then, we are ready to prove Theorem~\ref{thm:mec_exputil}:

\begin{proof}[Proof of Theorem~\ref{thm:mec_exputil}]
	To prove Theorem~\ref{thm:mec_exputil}, similar to the proof of Theorem~\ref{thm:single_buyer_regular}, we need to construct a set of feasible dual variables that makes the value of the dual program \eqref{eq:mech_dual} equals to the revenue of the optimal take-it-or-leave-it randomized price.
	
	First, consider the $\lambda$ and $\mu$ that we have used in the proof of Theorem~\ref{thm:single_buyer_regular}:
	For any $k,k' \in [K]$, we set $\mu_k=0$ and 
	\[
	\lambda_{kk'} = 
	\begin{cases}
	\sum_{\ell \ge k} f(v_{\ell}) z_M \frac{1}{u(v_k)-u(v_k-z_M)}, &\text{if }  k' = k-1\\
	0, &\text{otherwise}.
	\end{cases}
	\]
	We have already seen in the proof of Theorem~\ref{thm:single_buyer_regular} that 
	\[
	\Gamma_k(0;\lambda,\mu) = \Gamma_k(z_M;\lambda,\mu) = f(v_k) \phi_u(k) z_M.
	\]
	Then, we can apply the ironing process by adding \emph{loops} in $\lambda$ using \eqref{eq:loop1} and \eqref{eq:loop2}.  The loops are added in a way such that 
	\[
	\Gamma_k(0;\widetilde{\lambda},\mu) = f(v_k) \widetilde{\phi}_u(k) z_M
	\]
	for any $k \in [K]$, where $\widetilde{\lambda}$ is the $\lambda$ after the ironing process.  This can always be done by iteratively adding a properly chosen loop between each adjacent pairs of types within the same ironing interval.
	
	In the next step, we set $\widetilde{\nu}_k = \max\{0,f(v_k)\widetilde{\phi}_u(k)z_M\}$ and $\widetilde{\mu}_k = \mu_k$ for any $k \in [K]$.  We need to verify that the set of dual variables $(\widetilde{\lambda},\widetilde{\mu},\widetilde{\nu})$ is a set of feasible dual variables.  First, by definition of $\widetilde{\lambda}$, we have
	\[
	\Gamma_k(0;\widetilde{\lambda},\widetilde{\mu}) \le \widetilde{\nu}_k.
	\]
	for any $k \in [K]$.  Second, by the property of adding a loop, we have
	\begin{align*}
	\Gamma_k(z_M;\widetilde{\lambda},\widetilde{\mu}) &= \Gamma_k(0;\widetilde{\lambda},\widetilde{\mu}) \le \widetilde{\nu}_k \\
	\Pi_k(z_M;\widetilde{\lambda},\widetilde{\mu}) &= (1-e^{-\alpha z_M})\Gamma_k(0;\widetilde{\lambda},\widetilde{\mu}) \le \widetilde{\nu}_k
	\end{align*}
	Therefore, $(\widetilde{\lambda},\widetilde{\mu},\widetilde{\nu})$ is feasible.  Finally, under this choice of the dual variables and the definition of the ironed virtual value, we have
	\[
	\sum_{k} \nu_k = \sum_{k:\widetilde{\phi}_u(k)>0} f(v_k) \widetilde{\phi}_u(k) z_M.
	\]
	Since adding a loop within an interval does not change the average virtual value of the interval, we have
	\begin{align*}
	\sum_{k} \nu_k &= \sum_{k:\widetilde{\phi}_u(k)>0} f(v_k) \widetilde{\phi}_u(k) z_M \\
	&= \sum_{k:\widetilde{\phi}_u(k)>0} f(v_k) \phi_u(k) z_M  \\
	&= \sum_{k:\widetilde{\phi}_u(k)>0} \left(\frac{ \sum_{k' \ge k}f(v_{k'}) u(v_k)}{u(v_k)-u(v_k-z_M)} -  \frac{\sum_{k' \ge {k+1}}f(v_{k'}) u(v_{k+1})}{u(v_{k+1})-u(v_{k+1}-z_M)} \right) z_M  \\
	&= \sum_{k:\widetilde{\phi}_u(k)>0} f(v_{k})z_M \cdot \frac{u(v_{k^*})}{u(v_{k^*})-u(v_{k^*}-z_M)},
	\end{align*}
	where $k^* = \min\{k:\widetilde{\phi}_u(k)>0\}$.  This concludes the proof.
\end{proof}

\section{Missing Proofs from Optimal Mechanism Design for Multiple Buyers} \label{sec:app_multi_buyer}

\subsection{Proof of Theorem~\ref{lem:multibuyer_mech_feasible}}
\label{sec:app_multi_buyer_ppt}
First, we note that for each buyer $i \in [n]$ and $k \ge k^*$, the probability  $q_{i}(k)$ that she pays $z_M$ can be alternatively defined using the following recursive relationship
$$
q_i(k) = \frac{1}{1-x_i(k)} \cdot \left[ \frac{x_i(k) u(v_{k}) -  U_{i,k-1}(v_k)}{-u(-z_M)} \right],
$$
where $U_{i,k}(v)=x_i(k)u(v)+(1-x_i(k))q_i(k)u(-z_M)$ is the utility for buyer $i$ given she bids $v_{k}$ and her true value is $v$.  

We prove Theorem~\ref{lem:multibuyer_mech_feasible} by showing the following three lemmas:
\begin{lemma} \label{lem:2buyer_mech_ir}
	Consider $n \ge 2$ buyers with exponential utility.  With Assumption~(A1), the loser-pay auction is individually rational, i.e. for any buyer $i \in [n]$ and for any $k \in [K]$, it holds that $U_{i,k}(v_k) \ge 0$.
\end{lemma}
\begin{proof}
	We prove this lemma by induction for each buyer.  Consider buyer $i$. For $k < k^*$, we have that $U_{i,k}(v_k)=0$ by definition of the mechanism.  For $k = k^*$, we can see that
	$$
	U_{i,k^*}(v_{k^*}) = x_i(k^*) u(v_{k^*}) - x_i(k^*) u(v_{k^*}) = 0.
	$$
	For the induction step, consider any $k > k^*$.  Assume $U_{i,k-1}(v_{k-1}) \ge 0$ holds.  We first note that
	\begin{align*}
	U_{i,k}(v_k) &= x_i(k)u(v_k) + (1-x_i(k))q_i(k)u(-z_M) \\
	&= x_i(k)u(v_k) - \left[ x_i(k) u(v_{k}) - U_{i,k-1}(v_k) \right]\\
	&= U_{i,k-1}(v_k).
	\end{align*}
	Since $U_{i,k}(v)$ is monotone non-decreasing in $v$ for any $i \in [n]$ and any $k \in [K]$, we have $U_{i,k-1}(v_k) \ge U_{i,k-1}(v_{k-1})$.  Therefore, by the induction hypothesis, we have
	$$ U_{i,k}(v_k) \ge U_{i,k-1}(v_{k-1}) \ge 0. $$
	The lemma follows by induction.
\end{proof}

\begin{lemma} \label{lem:2buyer_mech_feasible}
	Consider $n \ge 2$ buyers with exponential utility. With Assumption~(A1), the loser-pay auction is feasible, i.e. for each buyer $i \in [n]$ and for any $k \in [K]$, it holds that $q_i(k) \in [0,1]$.
\end{lemma}
\begin{proof}
	We prove this by induction for each buyer. Consider buyer $i$. For $k < k^*$, by definition of $q_i(k)$, we have $q_i(k)=0$.  For $k=k^*$, it is clear that $q_i(k^*) \ge 0$.  Also, by assumption (A1), we have 
	$$
	q_i(k^*) = \frac{x_i(k^*)}{1-x_i(k^*)} \cdot \frac{u(v_{k^*})}{-u(-z_M)} \le \frac{1-\frac{1}{n}f(v_{k^*})}{\frac{1}{n}f(v_{k^*})} \cdot \frac{e^{\alpha v_{k^*}}-1}{e^{-\alpha z_M}-1} < 1.
	$$
	For the induction step, consider any $k > k^*$ and assume $q_i(k-1) \in [0,1]$.  First, to show that $q_i(k) \ge 0$, note that
	\begin{align*}
	x_i(k)u(v_k) - U_{i,k-1}(v_k) &= x_i(k)u(v_k) - x_i(k-1)u(v_k) - (1-x_i(k-1))q_i(k-1)u(-z_M) \\
	&= \left( x_i(k)-x_i(k-1) \right) u(v_k) + \left( 1-x_i(k-1) \right) q_i(k-1) (-u(-z_M)) \\
	&\ge 0,
	\end{align*}
	where the last inequality holds because $x_i(k) \ge x_i(k-1)$ by definition, as well as the induction hypothesis.  For the upper bound, since by Lemma~\ref{lem:2buyer_mech_ir}, we have $U_{i,k-1}(v_k) \ge 0$ for any $k > k^*$,  therefore, 
	$$
	q_i(k) < \frac{x_i(k)}{1-x_i(k)} \cdot \frac{u(v_k)}{-u(-z_M)} \le \frac{1-\frac{1}{n}f(v_{k})}{\frac{1}{n}f(v_{k})} \cdot \frac{e^{\alpha v_{k}}-1}{e^{-\alpha z_M}-1} < 1.
	$$
	The lemma follows by induction.
\end{proof}

\begin{lemma}
	Consider $n \ge 2$ buyers with exponential utility. With Assumption~(A1), the loser-pay auction is Bayesian incentive compatible.
\end{lemma}
\begin{proof}
	Fix any $k \in [K]$. We would like to show that for any $k' \in [K]$ and $k' \not=k$, it holds that $U_{i,k}(v_k) \ge U_{i,k'}(v_k)$.  Before we prove the lemma, note that
	$$
	\frac{\partial U_{i,k}(v)}{\partial v} = \alpha x_i(k) e^{\alpha v}.
	$$
	Hence, by definition of $x_i(k)$, for any $k > k'$, we have $x_i(k) \ge x_i(k')$.  Therefore, for any $k > k'$, we have $\frac{\partial U_{i,k}(v)}{\partial v} \ge \frac{\partial U_{i,k'}(v)}{\partial v}$.  
	
	To prove the lemma, first, we show that for any $k' < k$, we have $U_{i,k}(v_k) \ge U_{i,k'}(v_k)$. To see this, we claim that $U_{i,k}(v) \ge U_{i,k'}(v)$ for any $v \ge v_k$.  We prove this by induction on $k'$.  The base case ($k'=k-1$) holds since $U_{i,k}(v_k)=U_{i,k-1}(v_k)$ by definition of $q_i(k)$. For the induction step, consider any $k' < k-1$ and assume $U_{i,k}(v) \ge U_{i,k'+1}(v)$ holds for any $v \ge v_k$.  Since $U_{i,k'+1}(v_{k'+1})=U_{i,k'}(v_{k'+1})$ by definition of $q_i(k'+1)$, we can find that $U_{i,k'+1}(v) \ge U_{i,k'}(v)$ for any $v \ge v_{k'+1}$.  Therefore, we have $U_{i,k}(v) \ge U_{i,k'+1}(v) \ge U_{i,k'}(v)$ for any $v \ge v_k$.  The claim follows by induction.
	
	For $k' > k$, we will use a similar argument, i.e. we claim that $U_{i,k}(v) \ge U_{i,k'}(v)$ for any $v \le v_k$ and prove it by induction.  The base case ($k'=k+1$) holds because $U_{i,k}(v_{k+1}) = U_{i,k+1}(v_{k+1})$.  For the induction step, consider any $k' > k+1$ and assume $U_{i,k}(v) \ge U_{i,k'-1}(v)$ holds for any $v \le v_k$.  Since $U_{i,k'-1}(v_{k'}) = U_{i,k'}(v_{k'})$, we have $U_{i,k'-1}(v) \ge U_{i,k'}(v)$ for any $v \le v_{k'}$.  Therefore, we have $U_{i,k}(v) \ge U_{i,k'-1}(v) \ge U_{i,k'}(v)$ for any $v \le v_k$.  The claim follows by induction.  This completes the proof of the lemma.	 
\end{proof}

\subsection{Duality Theory for Optimal Mechanism}
\label{sec:app_multi_buyer_opt}

\begin{lemma} \label{lemma:two_buyer_exp_two_price1}
	In the dual program \eqref{eq:mech_dual_multibuyer}, both $\Gamma_{i,k}(z;\lambda,\mu)$ and $\Pi_{i,k}(z;\lambda,\mu)$ are either increasing or strongly convex in $z$ for any $i \in [n]$ and for $z \ge 0$.
\end{lemma}
\begin{proof}
	To show the lemma, we can simply rewrite $\Gamma_{i,k}(z;\lambda,\mu)$ as
	\begin{align}
	\Gamma_{i,k}(z;\lambda,\mu) =f(v_k)z + \beta B_{i,k} e^{-\alpha z} - \beta A_{i,k},
	\end{align}
	where
	\begin{align*}
	A_{i,k} &= \sum_{k'}(\lambda_{i,k,k'}-\lambda_{i,k',k}) + \mu_{i,k} \\
	B_{i,k} &= \sum_{k'}(\lambda_{i,k,k'}e^{\alpha v_k} - \lambda_{i,k',k}e^{\alpha v_{k'}}) +  \mu_{i,k} e^{\alpha v_k}.
	\end{align*}
	We can find that if $B_{i,k}>0$, then $\Gamma_{i,k}(z;\lambda,\mu)$ is strongly convex in $z$ whereas if $B_{i,k} \le 0$, then $\Gamma_{i,k}(z;\lambda,\mu)$ is monotone increasing in $z$.
	Similarly, for $\Pi_{i,k}(z;\lambda,\mu)$, we have
	\begin{align}
	\Pi_{i,k}(z;\lambda,\mu) =f(v_k)z + \beta A_{i,k} e^{-\alpha z} - \beta A_{i,k}.
	\end{align}
	Also, we can find that if $A_{i,k}>0$, then $\Pi_{i,k}(z;\lambda,\mu)$ is strongly convex in $z$ whereas if $A_{i,k} \le 0$, then $\Pi_{i,k}(z;\lambda,\mu)$ is monotone increasing in $z$.
\end{proof}

\begin{proof}[Proof of Theorem~\ref{thm:two_buyer_regular}, with regularity condition]
	To prove the theorem, we construct a set of feasible dual variables and show that it makes the objective of the dual program \eqref{eq:mech_dual_multibuyer} equal to the revenue of the proposed mechanism.
	
	For each buyer $i \in [n]$ and any $k,k' \in [K]$, we set $\mu_{i,k}=0$ and 
	\[
	\lambda_{i,k,k'} = 
	\begin{cases}
	\sum_{\ell \ge k} f(v_{\ell}) z_M \frac{e^{\alpha v_k}}{u(v_k)-u(v_k-z_M)}, &\text{if }  k' = k-1\\
	0, &\text{otherwise}.
	\end{cases}
	\]
	Also, for each buyer $i \in [n]$ and any $\vv{k} \in [K]^n$, we set $\nu_{i,\vv{k}}=0$ and
	\[
	\gamma_{\vv{k}} = \max_{i \in [n]} \left\{ 0, f(\vv{k})\Phi_u(k_i)z_M  \right\}.
	\]
	
	We first verify that this choice of values for the dual variables is feasible for the dual program \eqref{eq:mech_dual_multibuyer}.  By Lemma~\ref{lemma:two_buyer_exp_two_price1}, it suffices to check whether $f(\vv{k}_{-i})\Gamma_{i,k_i}(0;\lambda,\mu) \le \gamma_{\vv{k}}$, $f(\vv{k}_{-i})\Gamma_{i,k_i}(z_M;\lambda,\mu) \le \gamma_{\vv{k}}$, and $f(\vv{k}_{-i})\Pi_{i,k_i}(z_M;\lambda,\mu) \le 0$. 
	First, we have
	\begin{align*}
	f(\vv{k}_{-i})\Gamma_{i,k_i}(0; \lambda,\mu) &= f(\vv{k}_{-i}) \left[ \lambda_{i,k_i,k_i-1} u(v_{k_i}) - \lambda_{i,{k_i}+1,k_i} u(v_{k_i+1}) \right]\\
	&= f(\vv{k}_{-i}) \left[ \frac{\sum_{\ell \ge k_i} f(v_{\ell}) z_M e^{\alpha v_{k_i}} u(v_{k_i})}{u(v_{k_i})-u(v_{k_i}-z_M)} - \frac{\sum_{\ell \ge k_i+1} f(v_{\ell}) z_M e^{\alpha v_{k_i+1}} u(v_{k_i+1})}{u(v_{k_i+1})-u(v_{k_i+1}-z_M)} \right]\\
	&= f(\vv{k})\Phi_u(k_i)z_M \le \gamma_{\vv{k}}
	\end{align*}
	and 
	\begin{align*}
	&f(\vv{k}_{-i})\Pi_{i,k_i}(z_M;\lambda,\mu) \\
	&\quad = f(\vv{k}_{-i})\left[ f(v_{k_i}) z_M + (\lambda_{i,k_i,k_i-1} - \lambda_{i,k_i+1,k_i}) u(-z_M) \right] \\
	&\quad = f(\vv{k}_{-i})\left[ f(v_{k_i}) + \left( \frac{\sum_{\ell \ge k_i} f(v_{\ell}) e^{\alpha v_{k_i}}}{u(v_{k_i})-u(v_{k_i}-z_M)} -  \frac{\sum_{\ell \ge k_i+1} f(v_{\ell}) e^{\alpha v_{k_i+1}}}{u(v_{k_i+1})-u(v_{k_i+1}-z_M)} \right) u(-z_M) \right] z_M \\
	&\quad = f(\vv{k}_{-i})\left[ f(v_{k_i}) + \left( \frac{\sum_{\ell \ge k_i} f(v_{\ell}) e^{\alpha v_{k_i}}}{e^{\alpha v_{k_i}}(1-e^{-\alpha z_M})} -  \frac{\sum_{\ell \ge k_i+1} f(v_{\ell}) e^{\alpha v_{k_i+1}}}{e^{\alpha v_{k_i+1}}(1-e^{-\alpha z_M})} \right) (e^{-\alpha z_M}-1) \right] z_M \\
	&\quad = 0.
	\end{align*}
	For $\Gamma_{i,k_i}(z_M;\lambda,\mu)$, we note that
	\begin{align*}
	f(v_k)\Phi_u(k)z_M &= \sum_{\ell \ge k} f(v_{\ell}) z_M \frac{e^{\alpha v_k} u(v_k)}{u(v_k)-u(v_k-z_M)} - \sum_{\ell \ge k+1} f(v_{\ell}) z_M \frac{e^{\alpha v_{k+1}} u(v_{k+1})}{u(v_{k+1})-u(v_{k+1}-z_M)}\\
	&= \sum_{\ell \ge k} f(v_{\ell}) z_M \frac{e^{\alpha v_k}-1}{1-e^{-\alpha z_M}} - \sum_{\ell \ge k+1} f(v_{\ell}) z_M \frac{e^{\alpha v_{k+1}}-1}{1-e^{-\alpha z_M}}\\
	&= -\frac{f(v_k)}{1-e^{-\alpha z_M}} + \sum_{\ell \ge k} f(v_{\ell}) z_M \frac{e^{\alpha v_{k}}}{1-e^{-\alpha z_M}} - \sum_{\ell \ge k+1} f(v_{\ell}) z_M \frac{e^{\alpha v_{k+1}}}{1-e^{-\alpha z_M}}.
	\end{align*}
	Then, we have
	\begin{align*}
	&f(\vv{k}_{-i})\Gamma_{i,k_i}(z_M; \lambda,\mu) \\
	&\quad= f(\vv{k}_{-i}) \left[ f(v_{k_i})z_M + \lambda_{i,k_i,k_i-1} u(v_{k_i}-z_M) - \lambda_{i,k_i+1,k_i} u(v_{k_i+1}-z_M) \right]\\
	&\quad= f(\vv{k}_{-i}) \left[ f(v_{k_i}) + \frac{\sum_{\ell \ge k_i} f(v_{\ell})  e^{\alpha v_{k_i}} u(v_{k_i}-z_M)}{u(v_{k_i})-u(v_{k_i}-z_M)} -  \frac{\sum_{\ell \ge k+1} f(v_{\ell}) e^{\alpha v_{k_i+1}} u(v_{k_i+1}-z_M)}{u(v_{k_i+1})-u(v_{k_i+1}-z_M)} \right] z_M\\
	&\quad= f(\vv{k}_{-i}) \left[ f(v_{k_i}) + \frac{\sum_{\ell \ge k_i} f(v_{\ell})  e^{\alpha v_{k_i}} (e^{\alpha (v_{k_i}-z_M)} - 1)}{e^{\alpha v_{k_i}} ( 1- e^{-\alpha z_M})} -  \frac{\sum_{\ell \ge k+1} f(v_{\ell}) e^{\alpha v_{k_i+1}} (e^{\alpha (v_{k_i+1}-z_M)}-1)}{e^{\alpha v_{k_i+1}}(1-e^{-\alpha z_M})} \right] z_M\\
	&\quad= f(\vv{k}_{-i}) \left[ f(v_{k_i})\frac{-e^{-\alpha z_M}}{1-e^{-\alpha z_M}} + \frac{\sum_{\ell \ge k_i} f(v_{\ell})  e^{\alpha (v_{k_i}-z_M)}}{1- e^{-\alpha z_M}} -  \frac{\sum_{\ell \ge k+1} f(v_{\ell}) e^{\alpha (v_{k_i+1}-z_M)}}{1-e^{-\alpha z_M}} \right] z_M\\
	&\quad= f(\vv{k})\Phi_u(k_i)z_M e^{-\alpha z_M} \le \gamma_{\vv{k}}.
	\end{align*}
	
	For the next step, we show that under this choice of dual variables, the dual objective equals the revenue of the mechanism.  The revenue of the mechanism is the sum of the payments from all buyers.  Consider buyer $i$. The expected payment by buyer $i$ can be written as
	\begin{align}
	&\sum_{k_i \ge k^*} f(v_{k_i})(1-x_i(k_i))q_i(k_i)z_M \nonumber\\
	&\quad= \sum_{k_i \ge k^*} \sum_{k = k^*}^{k_i} f(v_{k_i})\left[\frac{[x_i(k)-x_i(k-1)]u(v_{k})}{-u(-z_M)}\right] z_M \nonumber\\
	&\quad= \sum_{k_i \ge k^*} \left[\sum_{k \ge k_i} f(v_k) \frac{x_i(k_i)u(v_{k_i})}{-u(z_M)} - \sum_{k \ge k_i+1}f(v_k)\frac{x_i(k_i)u(v_{k_i+1})}{-u(z_M)}\right] z_M, \label{eq:two_buyers_rev_from_b1_}
	\end{align}
	where in the last equality we change the order of summation.
	
	Further, we note that
	\begin{align*}
	\sum_{\vv{k} \in [K]^n}\gamma_{\vv{k}} \left[ \frac{\mathbbm{1}\left\{ v_{k_i} \ge v_{k_{i'}}, \forall i' \not= i \right\}}{\sum_{i' \in [n]} \mathbbm{1}\left\{ v_{k_i} = v_{k_{i'}} \right\}} \right] 
	&= \sum_{k_i \ge k^*} \sum_{\vv{k}_{-i} \in [K]^{n-1}} \gamma_{\vv{k}} \left[ \frac{\mathbbm{1}\left\{ v_{k_i} \ge v_{k_{i'}}, \forall i' \not= i \right\}}{\sum_{i' \in [n]} \mathbbm{1}\left\{ v_{k_i} = v_{k_{i'}} \right\}} \right] \\
	&= \sum_{k_i \ge k^*} \sum_{\vv{k}_{-i} \in [K]^{n-1}} f(k_i,\vv{k}_{-i}) \Phi_u(k_i) z_M \left[ \frac{\mathbbm{1}\left\{ v_{k_i} \ge v_{k_{i'}}, \forall i' \not= i \right\}}{\sum_{i' \in [n]} \mathbbm{1}\left\{ v_{k_i} = v_{k_{i'}} \right\}} \right] \\
	&= \sum_{k_i \ge k^*} x_i(k_i) f(v_{k_i}) \Phi_u(k_i) z_M.
	\end{align*}
	
	Then, by definition of $\Phi_u(k_i)$, we have
	\begin{align}
	&\sum_{\vv{k} \in [K]^n}\gamma_{\vv{k}} \left[ \frac{\mathbbm{1}\left\{ v_{k_i} \ge v_{k_{i'}}, \forall i' \not= i \right\}}{\sum_{i' \in [n]} \mathbbm{1}\left\{ v_{k_i} = v_{k_{i'}} \right\}} \right] = \sum_{k_i \ge k^*} x_i(k_i) f(v_{k_i}) \Phi_u(k_i) z_M \nonumber\\
	&= \sum_{k_i \ge k^*} x_i(k_i) \left( \sum_{k_i' \ge k_i}f(v_{k_i'}) \cdot \frac{e^{\alpha v_{k_i}} u(v_{k_i})}{u(v_{k_i})-u(v_{k_i}-z_M)} - \sum_{k_i' \ge {k_i+1}}f(v_{k_i'}) \cdot \frac{e^{\alpha v_{k_i+1}} u(v_{k_i+1})}{u(v_{k_i+1})-u(v_{k_i+1}-z_M)} \right) z_M \nonumber\\
	&= \sum_{k_i \ge k^*} x_i(k_i) \left( \sum_{k_i' \ge k_i}f(v_{k_i'}) \cdot \frac{u(v_{k_i})}{-u(-z_M)} - \sum_{k_i' \ge {k_i+1}}f(v_{k_i'}) \cdot \frac{ u(v_{k_i+1})}{-u(-z_M)} \right) z_M, \label{eq:temp_1234_}
	\end{align}
	where in the last equality we use the fact that $u(v)-u(v-z_M)=\beta e^{\alpha v}(1-e^{-\alpha z_M}) = -e^{\alpha v}u(-z_M)$.  We can find that \eqref{eq:temp_1234_} is exactly equal to \eqref{eq:two_buyers_rev_from_b1_}.  Hence,
	$$ \sum_{\vv{k} \in [K]^n}\gamma_{\vv{k}} \left[ \frac{\mathbbm{1}\left\{ v_{k_i} \ge v_{k_{i'}}, \forall i' \not= i \right\}}{\sum_{i' \in [n]} \mathbbm{1}\left\{ v_{k_i} = v_{k_{i'}} \right\}} \right] = \sum_{k_i \ge k^*} f(v_{k_i})(1-x_i(k_i))q_i(k_i)z_M.$$
	Combining these together, we can conclude that the dual objective
	\begin{align*}
	\sum_{\vv{k} \in [K]^n} \gamma_{\vv{k}} 
	&=  \sum_{\vv{k} \in [K]^n}\gamma_{\vv{k}} \left[ \sum_{i \in [n]} \frac{ \mathbbm{1}\left\{ v_{k_i} \ge v_{k_{i'}}, \forall i' \not= i \right\}}{\sum_{i' \in [n]} \mathbbm{1}\left\{ v_{k_i} = v_{k_{i'}} \right\}} \right] \\
	&= \sum_{i \in [n]} \sum_{k_i \ge k^*} f(v_{k_i})(1-x_i(k_i))q_i(k_i)z_M
	\end{align*}
	is equal to the revenue of the mechanism.  Therefore the mechanism is optimal.
\end{proof}

\subsection{Optimal Mechanism beyond Regularity Condition}
\label{sec:app_iron_multiple}
In the multiple buyer scenario, we can repeat the ironing process that we have discussed in Section~\ref{sec:app_iron_single} with the new virtual value function defined in Section~\ref{sec:multi_buyer_intro}.  We can iron the virtual value function by first taking the convex hull of the revenue curve, defined as
\[
F_{\Phi}(k) = \sum_{\ell \le k} \Phi_u(\ell).
\]
Then, we take the derivative of the revenue curve as the ironed virtual value function.  Formally, we can define the ironing process for the multi-buyer case as:
\begin{definition}[Ironed virtual value for multiple buyers]
	Given a virtual value function $\Phi_u$.  Let $\{[a_1,b_1],[a_2,b_2],\dots,[a_P,b_P]\}$ denote the minimum intervals that are not convex on the revenue curve $F_{\Phi}$.  The ironed virtual value function is defined as
	\[
	\widetilde{\Phi}_{u}(k) = 
	\begin{cases}
	\frac{ \sum_{\ell = a_p}^{b_p} f(v_{\ell}) \Phi_u(\ell) }{ \sum_{\ell = a_p}^{b_p} f(v_{\ell}) } , &\text{if }  k \in [a_p,b_p], \text{ for any } p \in [P]\\
	\Phi_u(\ell), &\text{otherwise}.
	\end{cases}
	\]
\end{definition}

Then, we show the optimality of the loser-pay auction without regularity condition.

\begin{proof}[Proof of Theorem~\ref{thm:two_buyer_regular}, without regularity condition]
	To prove the theorem, we will take the similar modification to the dual variables as we have done in the proof of Theorem~\ref{thm:mec_exputil}.  We will show that we can add \emph{loops} to the dual variables assigned in Section~\ref{sec:app_multi_buyer_opt} to apply the ironing process.
	
	First, consider the process of adding a \emph{loop} to $\lambda$.  Let $\lambda'$ coincides with $\lambda$ except at $k_i$ and $k_i'$ for some $k_i' < k_i$ and some $i \in [n]$:
	\begin{align*}
	\lambda_{i,k_i,k_i'}' &= \lambda_{i,k_i,k_i'} + \sum_{\ell \ge k_i} A z_M \frac{e^{\alpha v_{k_i}}}{u(v_{k_i}) - u(v_{k_i}-z_M)}\\
	\lambda_{i,k_i',k_i}' &= \lambda_{i,k_i',k_i} + \sum_{\ell \ge k_i} A z_M \frac{e^{\alpha v_{k_i'}}}{u(v_{k_i'}) - u(v_{k_i'}-z_M)}
	\end{align*}
	for some $A > 0$.  First, we can find that
	\begin{align*}
	&f(\vv{k}_{-i}) \Gamma_{i,k_i}(0;\lambda',\mu) \\
	&= f(\vv{k}_{-i}) \left[ \Gamma_{i,k_i}(0;\lambda,\mu) + \sum_{\ell \ge k_i} A z_M \frac{e^{\alpha v_{k_i}}u(v_{k_i})}{u(v_{k_i}) - u(v_{k_i}-z_M)} - \sum_{\ell \ge k_i} A z_M \frac{e^{\alpha v_{k_i'}}u(v_{k_i'})}{u(v_{k_i'}) - u(v_{k_i'}-z_M)}\right]\\
	&= f(\vv{k}_{-i}) \left[ \Gamma_{i,k_i}(0;\lambda,\mu) + (K-k_i+1) A z_M \frac{e^{\alpha v_{k_i}} - e^{\alpha v_{k_i'}}}{1 - e^{-\alpha z_M}}\right]\\
	& > f(\vv{k}_{-i})\Gamma_{i,k_i}(0;\lambda,\mu).
	\end{align*}
	Similarly, we have
	\begin{align*}
	&f(\vv{k}_{-i}) \Gamma_{i,k_i'}(0;\lambda',\mu) \\
	&= f(\vv{k}_{-i}) \left[ \Gamma_{i,k_i'}(0;\lambda,\mu) + \sum_{\ell \ge k_i} A z_M \frac{e^{\alpha v_{k_i'}}u(v_{k_i'})}{u(v_{k_i'}) - u(v_{k_i'}-z_M)} - \sum_{\ell \ge k_i} A z_M \frac{e^{\alpha v_{k_i}}u(v_{k_i})}{u(v_{k_i}) - u(v_{k_i}-z_M)}\right]\\
	&= f(\vv{k}_{-i}) \left[ \Gamma_{i,k_i'}(0;\lambda,\mu) + (K-k_i+1) A z_M \frac{e^{\alpha v_{k_i'}} - e^{\alpha v_{k_i}}}{1 - e^{-\alpha z_M}}\right]\\
	& < f(\vv{k}_{-i})\Gamma_{i,k_i'}(0;\lambda,\mu).
	\end{align*}
	We can see that after adding the loop, $\Gamma_{i,k_i}$ is increased by $(K-k_i+1) A z_M \frac{e^{\alpha v_{k_i}} - e^{\alpha v_{k_i'}}}{1 - e^{-\alpha z_M}}$ and $\Gamma_{i,k_i'}$ is decreased by $(K-k_i+1) A z_M \frac{e^{\alpha v_{k_i}} - e^{\alpha v_{k_i'}}}{1 - e^{-\alpha z_M}}$.  Due to the connection of $\Gamma_{i,k_i}$ and the virtual value function as we mentioned in Section~\ref{sec:app_multi_buyer_opt}, adding the loop is equivalent to modifying the virtual value function.  In addition, we can find that adding the loop does not change the average virtual value of $k_i$ and $k_i'$, i.e. $f(\vv{k}_{-i})\Gamma_{i,k_i}(0;\lambda',\mu) + f(\vv{k}_{-i})\Gamma_{i,k_i'}(0;\lambda',\mu) = f(\vv{k}_{-i})\Gamma_{i,k_i}(0;\lambda,\mu) + f(\vv{k}_{-i})\Gamma_{i,k_i'}(0;\lambda,\mu) = f(\vv{k}_{-i},k_i)\Phi_u(k_i) + f(\vv{k}_{-i},k_i')\Phi_u(k_i')$.  Then, we make the following claim:
	\begin{claim}
		Assume $f(\vv{k}_{-i})\Gamma_{i,k}(z_M;\lambda,\mu) = f(\vv{k}_{-i})\Gamma_{i,k}(0;\lambda,\mu) e^{-\alpha z_M}$ and $f(\vv{k}_{-i})\Pi_{i,k}(z_M;\lambda,\mu) = 0$ hold for $k=k_i$ and $k=k_i'$ for some $k_i'<k_i$ and some $i \in [n]$.  Suppose $\lambda'$ is formed by adding a loop between $k_i$ and $k_i'$ in $\lambda$. Then, it holds that
		\begin{align*}
		f(\vv{k}_{-i})\Gamma_{i,k}(z_M;\lambda',\mu) &= f(\vv{k}_{-i})\Gamma_{i,k}(0;\lambda',\mu) e^{-\alpha z_M} \\
		f(\vv{k}_{-i})\Pi_{i,k}(z_M;\lambda',\mu) &= 0
		\end{align*}
		for $k=k_i$ and $k=k_i'$.
	\end{claim}
	To prove this claim, we first note that
	\begin{align*}
	&f(\vv{k}_{-i})\Gamma_{i,k_i}(z_M;\lambda',\mu) \\
	&= f(\vv{k}_{-i})\left[ \Gamma_{i,k_i}(z_M;\lambda,\mu) + \sum_{\ell \ge k_i} A z_M \frac{e^{\alpha v_{k_i}}u(v_{k_i}-z_M)}{u(v_{k_i}) - u(v_{k_i}-z_M)} - \sum_{\ell \ge k_i} A z_M \frac{e^{\alpha v_{k_i'}}u(v_{k_i'}-z_M)}{u(v_{k_i'}) - u(v_{k_i'}-z_M)} \right] \\
	&= f(\vv{k}_{-i})\left[ \Gamma_{i,k_i}(z_M;\lambda,\mu) + \sum_{\ell \ge k_i} A z_M \frac{e^{\alpha v_{k_i}}(e^{\alpha v_{k_i}}e^{-\alpha z_M} - 1)}{e^{\alpha v_{k_i}}(1-e^{-\alpha z_M})} - \sum_{\ell \ge k_i} A z_M \frac{e^{\alpha v_{k_i'}}(e^{\alpha v_{k_i'}}e^{-\alpha z_M} - 1)}{e^{\alpha v_{k_i'}}(1-e^{-\alpha z_M})} \right] \\
	&= f(\vv{k}_{-i})\left[ \Gamma_{i,k_i}(z_M;\lambda,\mu) + \sum_{\ell \ge k_i} A z_M \frac{e^{\alpha v_{k_i}}e^{-\alpha z_M} - e^{\alpha v_{k_i'}}e^{-\alpha z_M} }{1-e^{-\alpha z_M}} \right] \\
	&= f(\vv{k}_{-i})\left[ \Gamma_{i,k_i}(0;\lambda,\mu) e^{-\alpha z_M} + (K-k_i+1) A z_M \frac{e^{\alpha v_{k_i}} - e^{\alpha v_{k_i'}}}{1-e^{-\alpha z_M}} e^{-\alpha z_M} \right] \\
	&= f(\vv{k}_{-i})\Gamma_{i,k_i}(0;\lambda',\mu)e^{-\alpha z_M}
	\end{align*}
	Using the similar derivation, we can also see that
	\begin{align*}
	&f(\vv{k}_{-i})\Gamma_{i,k_i'}(z_M;\lambda',\mu) \\
	&= f(\vv{k}_{-i})\left[ \Gamma_{i,k_i'}(z_M;\lambda,\mu) + \sum_{\ell \ge k_i} A z_M \frac{e^{\alpha v_{k_i'}}u(v_{k_i'}-z_M)}{u(v_{k_i'}) - u(v_{k_i'}-z_M)} - \sum_{\ell \ge k_i} A z_M \frac{e^{\alpha v_{k_i}}u(v_{k_i}-z_M)}{u(v_{k_i}) - u(v_{k_i}-z_M)} \right] \\
	&= f(\vv{k}_{-i})\left[ \Gamma_{i,k_i'}(0;\lambda,\mu) e^{-\alpha z_M} + (K-k_i+1) A z_M \frac{e^{\alpha v_{k_i'}} - e^{\alpha v_{k_i}}}{1-e^{-\alpha z_M}} e^{-\alpha z_M} \right] \\
	&= f(\vv{k}_{-i})\Gamma_{i,k_i'}(0;\lambda',\mu)e^{-\alpha z_M}
	\end{align*}
	For $f(\vv{k}_{-i})\Pi_{i,k}(z_M;\lambda',\mu)$, 
	\begin{align*}
	&f(\vv{k}_{-i})\Pi_{i,k_i}(z_M;\lambda',\mu) \\
	&= f(\vv{k}_{-i})\left[ \Pi_{i,k_i}(z_M;\lambda,\mu) + \sum_{\ell \ge k_i} A z_M \frac{e^{\alpha v_{k_i}}u(-z_M)}{u(v_{k_i}) - u(v_{k_i}-z_M)} - \sum_{\ell \ge k_i} A z_M \frac{e^{\alpha v_{k_i'}}u(-z_M)}{u(v_{k_i'}) - u(v_{k_i'}-z_M)} \right] \\
	&= f(\vv{k}_{-i})\left[ \Pi_{i,k_i}(z_M;\lambda,\mu) + \sum_{\ell \ge k_i} A z_M \frac{e^{\alpha v_{k_i}}(e^{-\alpha z_M} - 1)}{e^{\alpha v_{k_i}}(1-e^{-\alpha z_M})} - \sum_{\ell \ge k_i} A z_M \frac{e^{\alpha v_{k_i'}}(e^{-\alpha z_M} - 1)}{e^{\alpha v_{k_i'}}(1-e^{-\alpha z_M})} \right] \\
	&= f(\vv{k}_{-i})\Pi_{i,k_i}(z_M;\lambda,\mu)\\
	& = 0
	\end{align*}
	Similarly, we can conclude that $f(\vv{k}_{-i})\Pi_{i,k_i'}(z_M;\lambda',\mu)=0$.  This proves the claim.
	
	For the main part of the proof of the theorem, we first consider the assignment of the dual variables that we have used in Section~\ref{sec:app_multi_buyer_opt}: For each buyer $i \in [n]$ and any $k,k' \in [K]$, we set $\mu_{i,k}=0$ and 
	\[
	\lambda_{i,k,k'} = 
	\begin{cases}
	\sum_{\ell \ge k} f(v_{\ell}) z_M \frac{e^{\alpha v_k}}{u(v_k)-u(v_k-z_M)}, &\text{if }  k' = k-1\\
	0, &\text{otherwise}.
	\end{cases}
	\]
	From Section~\ref{sec:app_multi_buyer_opt}, we showed that this choice of $\lambda$ and $\mu$ has the following properties:
	\begin{enumerate}
		\item $f(\vv{k}_{-i})\Gamma_{i,k_i}(0;\lambda,\mu) = f(\vv{k})\Phi_u(k_i)z_M$ for any buyer $i \in [n]$ and $k_i \in [K]$.
		\item $f(\vv{k}_{-i})\Gamma_{i,k_i}(z_M;\lambda,\mu) = f(\vv{k})\Phi_u(k_i)z_M e^{-\alpha z_M}$ for any buyer $i \in [n]$ and $k_i \in [K]$.
		\item $f(\vv{k}_{-i})\Pi_{i,k_i}(z_M;\lambda,\mu) = 0$ for any buyer $i \in [n]$ and $k_i \in [K]$.
	\end{enumerate}
	Then, we apply the ironing process by adding \emph{loops} in $\lambda$.  The loops are added in a way such that 
	\[
	f(\vv{k}_{-i})\Gamma_{i,k_i}(0;\widetilde{\lambda},\mu) = f(\vv{k})\widetilde{\Phi}_u(k_i)z_M.
	\]
	for any buyer $i \in [n]$ and any $\vv{k} \in [K]^n$, where $\widetilde{\lambda}$ is the $\lambda$ after ironing.  This can always be done by iteratively adding a properly chosen loop between adjacent pairs of types within the same ironing interval.
	
	In the next step, for each buyer $i \in [n]$ and any $\vv{k} \in [K]^n$, we set $\nu_{i,\vv{k}} = 0$ and 
	\[
	\gamma_{\vv{k}} = \max_{i \in [n]} \{0,f(\vv{k})\widetilde{\Phi}_u(k_i)z_M\}.
	\]
	It is easy to see that the new assignment of the dual variables $(\widetilde{\lambda},\widetilde{\mu},\widetilde{\nu},\widetilde{\gamma})$ is a feasible assignment using the definition of $\lambda$ and the claim that we have made above.  Then, we note that
	\begin{align*}
	&\sum_{\vv{k} \in [K]^n}\widetilde{\gamma}_{\vv{k}} \left[ \frac{\mathbbm{1}\left\{ \widetilde{\Phi}_u(v_{k_i}) \ge \widetilde{\Phi}_u(v_{k_{i'}}), \forall i' \not= i \right\}}{\sum_{i' \in [n]} \mathbbm{1}\left\{ \widetilde{\Phi}_u(v_{k_i}) = \widetilde{\Phi}_u(v_{k_{i'}}) \right\}} \right] \\ 
	&= \sum_{k_i \ge k^*} \sum_{\vv{k}_{-i} \in [K]^{n-1}} \widetilde{\gamma}_{\vv{k}} \left[ \frac{\mathbbm{1}\left\{ \widetilde{\Phi}_u(v_{k_i}) \ge \widetilde{\Phi}_u(v_{k_{i'}}), \forall i' \not= i \right\}}{\sum_{i' \in [n]} \mathbbm{1}\left\{ \widetilde{\Phi}_u(v_{k_i}) = \widetilde{\Phi}_u(v_{k_{i'}}) \right\}} \right] \\
	&= \sum_{k_i \ge k^*} \sum_{\vv{k}_{-i} \in [K]^{n-1}} f(k_i,\vv{k}_{-i}) \widetilde{\Phi}_u(k_i) z_M \left[ \frac{\mathbbm{1}\left\{ \widetilde{\Phi}_u(v_{k_i}) \ge \widetilde{\Phi}_u(v_{k_{i'}}), \forall i' \not= i \right\}}{\sum_{i' \in [n]} \mathbbm{1}\left\{ \widetilde{\Phi}_u(v_{k_i}) = \widetilde{\Phi}_u(v_{k_{i'}}) \right\}} \right] \\
	&= \sum_{k_i \ge k^*} x_i(k_i) f(v_{k_i}) \widetilde{\Phi}_u(k_i) z_M.
	\end{align*}
	
	Then, by definition of $\widetilde{\Phi}_u(k_i)$, we have
	\begin{align}
	&\sum_{\vv{k} \in [K]^n} \widetilde{\gamma}_{\vv{k}} \left[ \frac{\mathbbm{1}\left\{ \widetilde{\Phi}_u(v_{k_i}) \ge \widetilde{\Phi}_u(v_{k_{i'}}), \forall i' \not= i \right\}}{\sum_{i' \in [n]} \mathbbm{1}\left\{ \widetilde{\Phi}_u(v_{k_i}) = \widetilde{\Phi}_u(v_{k_{i'}}) \right\}} \right] \nonumber \\
	&= \sum_{k_i \ge k^*} x_i(k_i) f(v_{k_i}) \widetilde{\Phi}_u(k_i) z_M
	= \sum_{k_i \ge k^*} x_i(k_i) f(v_{k_i}) \Phi_u(k_i) z_M \label{eq:temp_ironing1234}\\
	&= \sum_{k_i \ge k^*} x_i(k_i) \left( \sum_{k_i' \ge k_i}f(v_{k_i'}) \cdot \frac{e^{\alpha v_{k_i}} u(v_{k_i})}{u(v_{k_i})-u(v_{k_i}-z_M)} - \sum_{k_i' \ge {k_i+1}}f(v_{k_i'}) \cdot \frac{e^{\alpha v_{k_i+1}} u(v_{k_i+1})}{u(v_{k_i+1})-u(v_{k_i+1}-z_M)} \right) z_M \nonumber\\
	&= \sum_{k_i \ge k^*} x_i(k_i) \left( \sum_{k_i' \ge k_i}f(v_{k_i'}) \cdot \frac{u(v_{k_i})}{-u(-z_M)} - \sum_{k_i' \ge {k_i+1}}f(v_{k_i'}) \cdot \frac{ u(v_{k_i+1})}{-u(-z_M)} \right) z_M, \label{eq:temp_12345_}
	\end{align}
	where \eqref{eq:temp_ironing1234} is because adding loops within an interval do not change the average virtual value of the interval.  We can find that \eqref{eq:temp_12345_} exactly equals to \eqref{eq:two_buyers_rev_from_b1_}.  Therefore, we have
	\[
	 \sum_{k_i \ge k^*} f(v_{k_i})(1-x_i(k_i))q_i(k_i)z_M = \sum_{\vv{k} \in [K]^n} \widetilde{\gamma}_{\vv{k}} \left[ \frac{\mathbbm{1}\left\{ \widetilde{\Phi}_u(v_{k_i}) \ge \widetilde{\Phi}_u(v_{k_{i'}}), \forall i' \not= i \right\}}{\sum_{i' \in [n]} \mathbbm{1}\left\{ \widetilde{\Phi}_u(v_{k_i}) = \widetilde{\Phi}_u(v_{k_{i'}}) \right\}} \right] = \sum_{\vv{k}\in[K]^n}\widetilde{\gamma}_{\vv{k}},
	\]
	i.e. the revenue of the loser-pay auction is equal to the dual objective under the dual variables $(\widetilde{\lambda},\widetilde{\mu},\widetilde{\nu},\widetilde{\gamma})$.  Therefore, by strong duality, the mechanism is optimal.
\end{proof}

\section{Sub-optimality for General Utility Functions}\label{quadratic}

Consider the class of quadratic utility functions with the form $u(x)=\beta[(x+L)^2-L^2]$
for some $L>z_M$.  
Unlike exponential utility functions, for quadratic utility, the optimal mechanism can contain more than $2$ menu options.  We provide \evn{an example} to show the following theorem:

\begin{theorem} \label{thm:quadratic_util_failure}
	Consider a single risk-loving buyer. There exists a convex utility function $u$, and a distribution $f$ over a set of possible values $\setbvals$ such that the revenue maximizing randomized take-it-or-leave-it price is not optimal.
\end{theorem} 
\begin{proof}
	Consider the utility function $u(x)=(x+1)^2-1$. Assume the buyer's value is uniformly distributed between $\{0,0.1,0.2\dots,0.9\}$.  Also, assume $\setprices=\{0,1\}$.  The revenue of the optimal take-it-or-leave-it randomized pricing can be found by solving
	$$
	\max_{v \in \setbvals} \frac{\Pr[t\ge v] u(v)}{u(v)-u(v-z_M)}  \approx 0.3200.
	$$
	where the optimal revenue is attained when offering the two-priced menu option $(1,0.5333,0)$.
	
	Then, consider the mechanism with the following two-priced menu options:
	\begin{align*}
	&\left(\frac{1}{1.96},0,1\right) \approx (0.5102,0,1),
	&\left(1,\frac{1536}{2695},0\right) \approx (1,0.5699,0)
	\end{align*}
	For the first menu option, its utility curve is
	$$ U_1(v) = \frac{1}{1.96} u(v) + \frac{0.96}{1.96} u(-1) = \frac{1}{1.96}(v+1)^2 - 1. $$
	Therefore, for any buyer with $v \ge 0.4$, she prefers this menu option to $(0,0,0)$.  Next, the utility curve of the second menu option is 
	$$ U_2(v) = \frac{1536}{2695} u(v-1) + \frac{1159}{2695} u(v) = v^2 + \frac{2318}{2695}v - \frac{1536}{2695}.$$
	If we compare $U_1(v)$ with $U_2(v)$, we can find that
	$$ U_2(v) - U_1(v) = \frac{24}{2695}(11v+3)(5v-3), $$
	which means that for any buyer with $v \ge 0.6$, she prefers the second menu option to the first one.  The revenue of this mechanism is \evn{therefore}
	$$ \frac{0.96}{1.96}\cdot 0.2 + \frac{1536}{2695} \cdot 0.4 \approx 0.3259, $$
	which is greater than the optimal take-it-or-leave-it randomized pricing. 
\end{proof}


\end{document}